\documentclass[reprint,amsmath,amssymb,aps,superscriptaddress, nofootinbib]{revtex4-2}
\usepackage[colorlinks=true,allcolors=blue]{hyperref}

\usepackage{graphicx,dcolumn,booktabs,amsthm,xcolor,bm}
\usepackage[capitalize]{cleveref}

\newtheorem{thm}{Theorem}[section]
\newtheorem{proposition}[thm]{Proposition}
\newtheorem{lem}[thm]{Lemma}

\begin{document}

\preprint{APS/123-QED}

\title{Finite relaxation protocols with minimal dissipation}

\author{Ben Ansbacher}
\email{bdansbacher@gmail.com}
\affiliation{Carleton College, Northfield, MN 55057, USA}
\affiliation{Santa Fe Institute, Santa Fe, NM 87501, USA}

\author{Harrison Hartle}
\affiliation{Santa Fe Institute, Santa Fe, NM 87501, USA}

\author{Abhishek Yadav}
\affiliation{Santa Fe Institute, Santa Fe, NM 87501, USA}
\affiliation{University of New Mexico, Albuquerque, NM 87106, USA}
\affiliation{Complexity Science Hub, Metternichgasse 8, 1030, Vienna, Austria}

\author{Jan Korbel}
\affiliation{Complexity Science Hub, Metternichgasse 8, 1030, Vienna, Austria}

\author{David H. Wolpert}
\affiliation{Santa Fe Institute, Santa Fe, NM 87501, USA}
\affiliation{Complexity Science Hub, Metternichgasse 8, 1030, Vienna, Austria}

\date{\today}

\begin{abstract}
Work extraction from nonequilibrium
systems is a major challenge across biological, chemical, physical, and engineering systems. Idealized protocols generally require a quasistatic relaxation stage in which the Hamiltonian is gradually adjusted through a continuum of intermediaries. Here, we consider protocols restricted to a finite number $N$ of intermediary Hamiltonians, consisting of a sequence of quench-relax steps. We determine the sequence of quenches that minimizes the dissipated work, which can be expressed in terms of a recurrence involving the Lambert function. The optimal sequence converges to the Fisher-Rao geodesic, saturating known leading-order dissipation bounds at large $N$. We obtain lower bounds on work extraction from a nonequilibrium distribution as a function of its Fisher-Rao distance to equilibrium. We extend and apply the framework in two simple models: (i) an optical trap experiment, showing that the optimal intermediary distribution can be bimodal even for unimodal initial and final distributions, and (ii) an enzyme-catalyzed reaction, showing that that accounting for relaxation time in addition to dissipation can favor barrier-lowering.
\end{abstract}

\maketitle

\section{Introduction}
\label{sec:intro}

Extracting work from nonequilibrium systems is a fundamental problem in statistical physics, with applications ranging from molecular and biological systems to engineered devices~\cite{blaber2023optimal,brown2019theory,leighton2025flow}. For a system with Hamiltonian $H_0$ initially prepared in a nonequilibrium distribution $\rho_0$, the maximum extractable work is the excess nonequilibrium free energy of $\rho_0$~\cite{hasegawa2010generalization}. Given unrestricted Hamiltonian control, this bound can be saturated by an ideal protocol: the Hamiltonian is instantaneously quenched to an $H'$ for which $\rho_0$ is the equilibrium distribution, and then returned quasistatically to $H_0$~\cite{parrondo2015thermodynamics}. The quasistatic return is reversible but requires access to a continuum of intermediary Hamiltonians, which is a degree of control often unavailable in realistic work extraction protocols. 

Similar limitations have motivated the study of work extraction subject to protocol constraints~\cite{kolchinsky2021work}. Such constraints often pertain to the set of available quenching Hamiltonians, for instance, the ability to manipulate only certain energies \cite{Kolchinsky_2021}, having access to fewer quenches than possible nonequilibrium scenarios~\cite{Hartle2024}, or having only certain parameterized forms of quench as determined by experiment~\cite{Toyabe2010, Koski2014, Vidrighin2016}. But in many of the aforementiond cases, a quasistatic return from $H'$ to $H_0$ is assumed possible in the protocol, implicitly invoking a continuum of Hamiltonians. If we only have access to a finite collection of Hamiltonians, what is the best alternative to a quasistatic relaxation that can be carried out? In this work, we study work extraction when the return protocol consists of $N$ Hamiltonians, with complete equilibration following each quench. We refer to such protocols as \emph{finite relaxation protocols}. Similar sequential protocols or `step-equilibrations' have been used to approximate quasistatic processes~\cite{Nulton1985}; our focus is on determining and characterizing the optimal intermediaries at finite $N$.

Prior work has established lower bounds on dissipated work at large $N$ in terms of the Fisher-Rao (FR) distance between the initial and final distributions~\cite{Diosi2000}. Moreover, it has been shown at finite $N$ that inserting an additional intermediary can always reduce dissipation~\cite{Salamon2023}. The total dissipated work of a finite relaxation protocol is the sum of Kullback-Leibler (KL) divergences between successive intermediary equilibrium distributions~\cite{e27010042}. We obtain a coupled set of implicit equations that determine the optimal intermediaries in terms of the Lambert $W$ function, and which are amenable to efficient computation. The problem of minimization of a sequential sum of KL divergences was also addressed in Ref.~\cite{Pavlichin_2016}, which our results align with. We compute these solutions and examine their thermodynamic properties; their associated minimal dissipated work, $\mathcal{W}_N^{\mathrm{diss}}$, achieves the known $O(N^{-1})$ leading-order bound~\cite{Salamon1983}. We examine the mean dissipated work and its probability density for randomly drawn initial and final distributions, and also obtain an asymptotic lower bound on the maximal {\it extracted} work, denoted $\mathcal{W}^{\mathrm{ext}}_N$, in terms of FR distance. Our results complement a variety of related works on optimal protocols for work extraction and dissipation-minimization in classical and quantum thermodynamics~\cite{Abiuso_2020,mckeever2026finite}.

The rest of this paper is organized as follows. Sec.~\ref{sec:framework} describes the framework of finite relaxation protocols, and Sec.~\ref{sec:results} presents our results on the optimal sequence of intermediary Hamiltonians. Sec.~\ref{sec:extensions} covers two framework extensions for applications: an optical trapping potential, and a minimal model of enzymatic catalysis. We conclude in Sec.~\ref{sec:discussion}.

\section{Finite relaxation protocol}
\label{sec:framework}

Consider a system with microstates $x \in X$ and Hamiltonian $H(x)$, coupled to a single thermal reservoir with inverse temperature $\beta$. Its equilibrium (Boltzmann-Gibbs) distribution is
\begin{equation}
    \pi^H(x) := \frac{e^{-\beta H(x)}}{Z_H},
\end{equation}
Any other distribution $\rho \ne \pi^H(x)$ is a nonequilibrium distribution with respect to $H$ and $\beta$. The free energy $F_H(\rho): = \langle H\rangle_\rho - \beta^{-1} S(\rho)$ is minimized by $\pi^H$, with $\langle H\rangle_\rho=\sum_{x\in X}\rho(x)H(x)$ denoting the average energy of outcomes from $\rho$ and $S(\rho) :=-\sum_{x\in X}\rho(x)\log\rho(x)$ denoting the Shannon entropy. The excess free energy of a nonequilibrium distribution $\rho$ is 
\begin{equation}
\label{eq:DeltaF}
   \Delta F_H(\rho):= F_H(\rho) - F_H(\pi^H) = \beta^{-1} D(\rho \| \pi^H),
\end{equation}
with $D(\rho||\rho'):=\sum_{x\in X}\rho(x)\log\frac{\rho(x)}{\rho'(x)}$ the Kullback-Leibler divergence. Equilibration of $\rho$ under a fixed Hamiltonian $H$ thus dissipates work $\beta^{-1} D(\rho \| \pi^H)$~\cite{kawai2007dissipation,takara2010generalization}. In an ideal work-extraction protocol, one first carries out an instantaneous quench (Hamiltonian adjustment) $H_0\rightarrow H'$ to
\begin{equation}
\label{eq:Hprime}
    H'(x) = -\beta^{-1} \log \rho_0(x) + \mathrm{const},
\end{equation}
so that $\rho_0 = \pi^{H'}$. Then, quasistatically, one readjusts $H'$ back to $H_0$~\cite{parrondo2015thermodynamics}, passing through a continuum of intermediaries, yielding the maximal work extraction~\eqref{eq:DeltaF}. We instead restrict the return to $N$ intermediary Hamiltonians $H_1, \ldots, H_N$, finally reaching $H_{N+1}:=H_0$, with complete equilibration following each quench (see Fig.~\ref{fig:protocol}). Writing $\rho_i := \pi^{H_i}$ for $i=1,...,N$ and $\rho_{N+1}=\pi^{H_0}$, the total dissipated work is then
\begin{equation}
\label{eq_optimization_prob}
\mathcal{W}^{\mathrm{diss}} = \beta^{-1}\sum_{i = 1}^{N+1} D(\rho_{i-1} \| \rho_{i}).
\end{equation}
This protocol assumes an initial quench $H_0\rightarrow H'$ of the form~\eqref{eq:Hprime}, which is the first step of the optimal protocol with quasistatic relaxation~\cite{parrondo2015thermodynamics}; such an $H'$ exists for every full-support distribution $\rho_0$. However, a natural extension is to include the first quench in the finite set, seeking the optimal $N+1$ Hamiltonians. We show in Appendix~\ref{ap:initialquench} that in this scenario, the optimal first quench is \emph{not} of the form~\eqref{eq:Hprime}, but that the optimal collection follows a mild adjustment of the results we present below. Since every full-support distribution determines a Hamiltonian up to an additive constant via a relation of the form~\eqref{eq:Hprime}, the problem of obtaining the $\mathcal{W}^{\mathrm{diss}}$-minimizing set $\{H_i\}_{i=1}^N$ is equivalent to that of obtaining the set $\{\rho_i\}_{i=1}^N$ of intermediary distributions minimizing~\eqref{eq_optimization_prob}.

\section{Main results}
\label{sec:results}

We now characterize the optimal intermediary distributions that minimize $\mathcal{W}^{\mathrm{diss}}$ of Eq.~\eqref{eq_optimization_prob}, or equivalently the sum $\sum_{i = 1}^{N+1} D(\rho_{i-1}\| \rho_i)$, for fixed initial and final distributions $\rho_0$ and $\rho_{N+1} := \pi^{H_0}$. We assume $\pi^{H_0}$ has full support on a finite set of outcomes $X$. The sum of KL divergences is strictly convex in $(\rho_1,\ldots\rho_N)$ and diverges at the boundaries (when any of the distributions have at least one zero of probability), hence there is a unique minimizer; this follows from joint strict convexity of $D(p||q)$ on the interior of the unit simplex, and its divergence when $p$ is at a boundary. The problem of minimizing a sum of intermediary KL divergences has also been considered in Ref.~\cite{Pavlichin_2016}; as we show in Appendix~\ref{ap:lambertproof}, the set of optimal intermediary distributions satisfy
    \begin{equation}
    \label{eq:lambertsol}
        \rho_i(x) = \frac{\rho_{i-1}(x)}{W(c_i \rho_{i-1}(x) / \rho_{i+1}(x))}, \ \ \ \ i=1,...,N,
    \end{equation}
    where $W$ denotes the principal branch of the Lambert $W$ function and the constants $c_i > 0$ are determined by the normalization conditions $\sum_{x \in X} \rho_i(x) = 1$ for $i=1,...,N$. Moreover, $c_i$ is related to $\rho_{i}$ and $\rho_{i+1}$ by $c_i=\exp(-D(\rho_i||\rho_{i+1})+1)$~\cite{nielsen2013symmetrical}. Note that if $\rho_{0}(x')=0$ for any $x'\in X$, Eq.~\eqref{eq:lambertsol} at $i=1$ reduces to $\rho_1(x')=\rho_{2}(x')/c_1$ by continuity. The equations~\eqref{eq:lambertsol}, also obtained in Ref.~\cite{Pavlichin_2016}, constitute a coupled nonlinear system, lacking a closed-form expression for $\{\rho_i\}_{i=1}^N$. Nevertheless, the solutions can straightforwardly be computed numerically (Appendix~\ref{ap:algorithm}). The associated optimal intermediary Hamiltonians are 
\begin{equation} 
\label{eq:H_i_opt}
H_i(x) = -\beta^{-1}\log\rho_i(x)+\mathrm{const}, \ \ \ \ i=1,...,N.
\end{equation}
Fig.~\ref{fig:protocol} displays the optimal intermediaries~\eqref{eq:lambertsol} on the unit 2-simplex ($|X|=3$) for $N=1$ and $N=7$. The sequence converges towards the Fisher-Rao (FR) geodesic~\cite{Ay_2019} as $N$ grows~\cite{Pavlichin_2016}. The FR metric is central to the field of information geometry~\cite{Bettmann2025}, and relates to the thermodynamic length~\cite{Salamon1983,crooks2007measuring}, appearing in known dissipation bounds for finite-time thermodynamic transformations~\cite{crooks2007measuring,Ito2018,Ito_2020} as well as $N$-step protocols~\cite{Nulton1985}; see~\eqref{eq:FR_asymptotic} below. Thus, the optimal finite-$N$ sequence~\eqref{eq:lambertsol} is a discrete counterpart of the continuous FR geodesic, despite that intermediaries need not lie on exactly on the FR geodesic. See Appendix~\ref{ap:fr_background} for relevant background on the FR metric. Other thermodynamic bounds involving information geometry have been obtained in terms of the Wasserstein distance~\cite{panaretos2019statistical}, in the form of speed limit theorems~\cite{nakazato2021geometrical,dechant2022minimum,van2023topological}.

The the minimum dissipation achievable with $N$ optimal intermediaries is denoted 
\begin{equation}
\label{eq:WdissN}
    \beta\mathcal{W}^{\mathrm{diss}}_N := \min_{\rho_1,\ldots,\rho_N} \sum_{i=1}^{N+1} D(\rho_{i-1}\Vert\rho_i),
\end{equation}
as saturated by the optimal sequence~\eqref{eq:lambertsol}. The leading order behavior of $\mathcal{W}_N^{\mathrm{diss}}$ at large $N$ decays as $N^{-1}$ and achieves the established lower bound~\cite{Nulton1985,Diosi2000} of
\begin{equation}
\label{eq:FR_asymptotic}
    \mathcal{W}^{\mathrm{diss}}_N=\frac{d_{\mathrm{FR}}(\rho_0,\pi^{H_0})^2}{2\beta N}+O(N^{-2}),
\end{equation}
 where $d_{\mathrm{FR}}(\rho_0,\pi^{H_0})$ denotes the FR distance~\cite{Ay_2019} between $\rho_0$ and $\pi^{H_0}$; see Appendix~\ref{ap:fisherrao}. The work {\it extracted} under the protocol with $N$ optimal intermediaries is
\begin{equation}
\label{eq:Wext_def}
    \mathcal{W}^{\mathrm{ext}}_N := \Delta F_{H_0}(\rho_0) - \mathcal{W}^{\mathrm{diss}}_N.
\end{equation}
It can be shown that $\mathcal{W}^\mathrm{ext}_N$ is monotonically increasing with $N$~\cite{Pavlichin_2016}, and from~\eqref{eq:FR_asymptotic} it follows that $\mathcal{W}^\mathrm{ext}_N\rightarrow\Delta F_{H_0}(\rho_0)$ as $N\rightarrow\infty$. In Appendix~\ref{ap:Wext_lower} we obtain a lower bound on $\mathcal{W}^{\mathrm{ext}}_N$ entirely in terms of the FR distance; we derive, up to an $O(N^{-2})$ remainder,
\begin{equation}
\label{eq:Wext_lower_FR}
    \mathcal W^\mathrm{ext}_N \gtrsim \frac{d_{\mathrm{FR}}(\rho_0,\pi^{H_0})^2}{\beta}\left(\frac{1}{4}-\frac{1}{2N}\right).
\end{equation}

The case $N=1$ is tractable and provides the simplest nontrivial intermediary satisfying Eq.~\eqref{eq:lambertsol}. Namely, the optimal intermediary $\rho:=\rho_1$ is 
\begin{equation}
\label{eq:N1_solution}
    \rho(x) = \frac{\rho_0(x)}{W(c\rho_0(x) / \pi^{H_0}(x))},
\end{equation}
where $c >0$ is determined by normalization. Distributions of the form~\eqref{eq:N1_solution} have been referred to as Jeffreys centroids, in the context of the initial and final distributions being the arithmetic and geometric means of histograms~\cite{nielsen2013symmetrical}. We analyze the solution~\eqref{eq:N1_solution} in Appendix~\ref{ap:propsoflambert}; the associated work extraction obeys
\begin{equation}
\label{eq:Wext1}
\beta\mathcal{W}^{\mathrm{ext}}_1 = \left\langle W\left(\frac{c \rho_0(x)}{\pi^{H_0}(x)}\right) \right\rangle_{\rho_0} - 1\ge 0,
\end{equation}
with equality only if $\rho_0=\pi^{H_0}$. Dissipation can thus be reduced even by introducing just one intermediary, reflecting the `ladder property' of Ref.~\cite{Salamon2023}.

\begin{figure}
    \centering
    \includegraphics[width=1\linewidth]{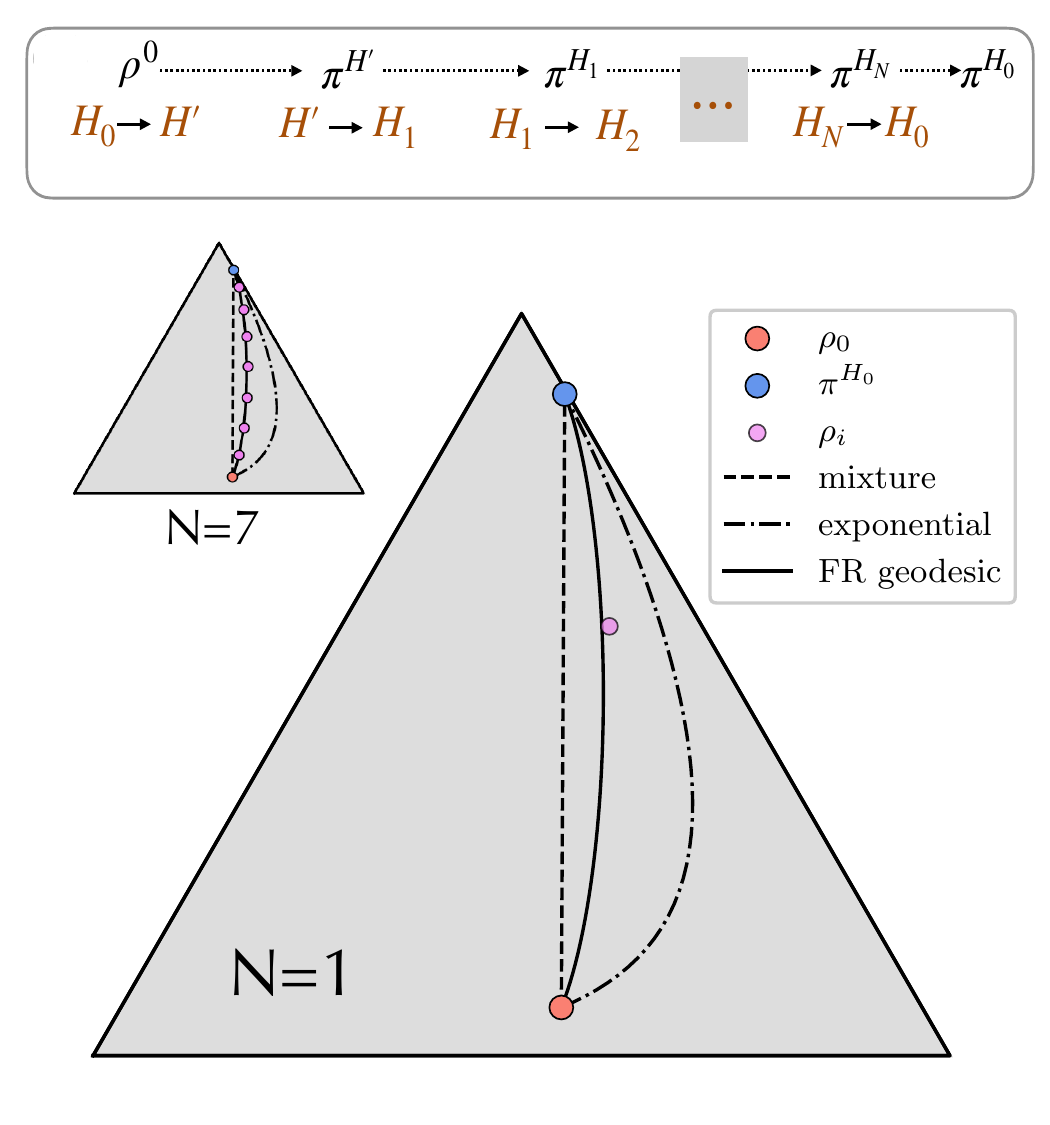}
    \caption{Top: Schematic of quench and relax protocol using finite many Hamiltonians. 
    Bottom: Plot on 2-simplex of the optimal intermediary $\rho_i$ (magenta) given by~\eqref{eq:lambertsol}. The initial distribution is $\rho_0$ and the final distribution is $\pi^\beta$. $\rho_i$ is relatively close to the the Fisher-Rao geodesic, but not exactly on it. Moreover, $\rho_i$ does not fall on the linear mixture interpolation $(1-t)\rho_0(x) + t\pi^\beta(x)$ or the exponential interpolation $\rho_0(x)^t(\pi^\beta(x))^{1-t}/Z(t)$~\cite{Ay_2019}. Inset: Optimal intermediaries for $N=7$. The optimal intermediary distributions $\rho_i$ for $i=1,...,7$ satisfy~\eqref{eq:lambertsol}. We can see the optimal intermediary distributions converge to the FR geodesic for large $N$.}
    \label{fig:protocol}
\end{figure}

\begin{figure}[h!]
    \centering
\includegraphics[width=1.05\linewidth,trim=50 0 0 10,clip]{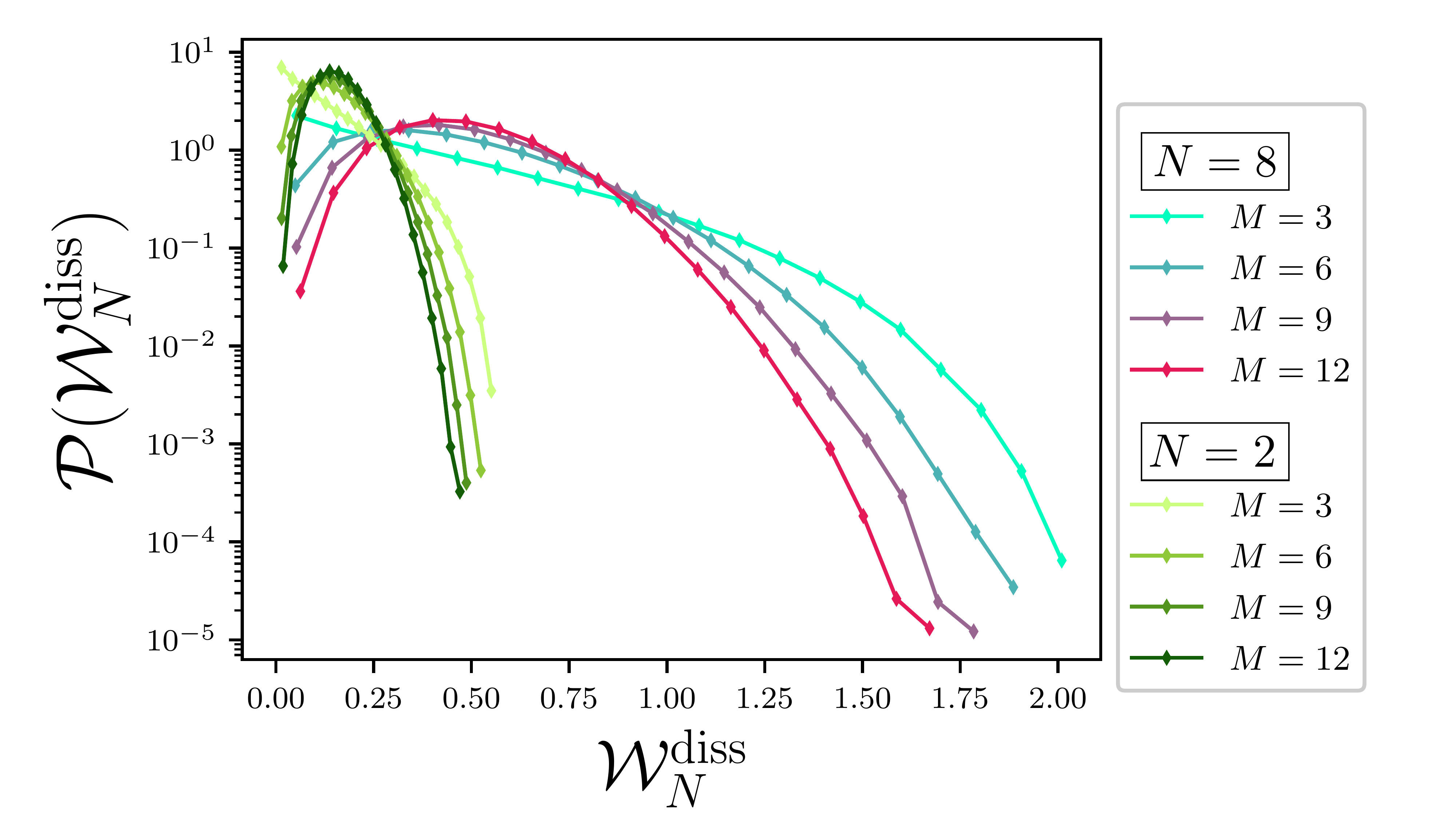}
\includegraphics[width=1.0\linewidth,trim=0 10 0 0,clip]{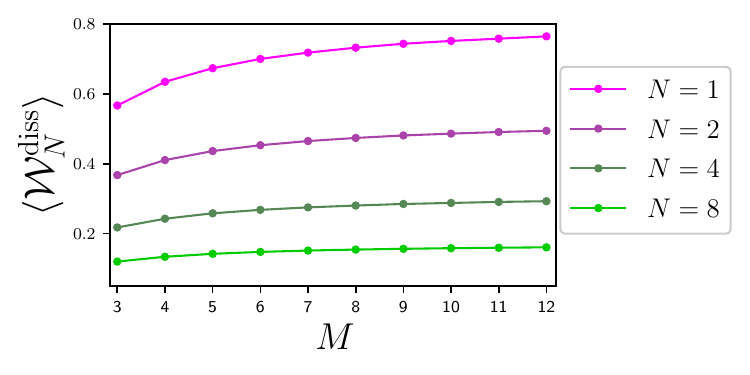}
    \caption{Minimal dissipated work $\mathcal{W}_N^{\mathrm{diss}}$ (Eq.~\eqref{eq:WdissN}) from $9\times 10^5$ pairs $(\rho^0,\pi^{H_0})$ drawn i.i.d. from $\mathrm{Dirichlet}(\frac{1}{2},...,\frac{1}{2})$ (Jeffreys prior). Upper panel: estimated probability density $\mathcal{P}(\mathcal{W}^{\mathrm{diss}}_N)$ at $N=2$ and $N=8$ across $M:=|X|$, using $20$ linearly-spaced bins with upper and lower bin edges set by the observed minimum and maximum values. Lower panel: at each $N$ as $M\rightarrow\infty$ the mean dissipated work $\langle \mathcal{W}_N^{\mathrm{diss}}\rangle$ saturates to a finite limit.}
    \label{fig:Wdiss}
\end{figure}

In Fig.~\ref{fig:Wdiss} we show the minimal dissipated work $\mathcal{W}_N^{\mathrm{diss}}$ of finite relaxation protocols with $N$ intermediaries and $M:=|X|$ outcomes, with initial and final distributions each drawn from the Jeffreys prior over categorical distributions, or equivalently the $\mathrm{Dirichlet}(\frac{1}{2},...,\frac{1}{2})$ density over the unit simplex~\cite{Jeffreys}, which is proportional to the square root of the determinant of the Fisher information matrix. Mean dissipated work appears to converge to finite limiting values as $M\rightarrow\infty$ at fixed $N$.

\section{Extensions}
\label{sec:extensions}

We now present two extensions of the framework of finite relaxation introduced above, arising in application to two simple models. The first example (Sec.~\ref{ssec:optical_trap}) is for an experimental setting---a particle in a 1D space manipulated by an optical trap subject to external tuning, requiring extension of our results in Sec.~\ref{sec:results} to a continuum sample space $X=\mathbb{R}$. The second example (Sec.~\ref{ssec:enzyme}) is an elementary model of enzymatic catalysis; we show that minimization of $\mathcal{W}^{\mathrm{diss}}$ alone fails to reproduce the barrier-lowering required for catalysis, motivating an adjusted optimization involving reaction timescale in addition to dissipated work.

\subsection{Optical trap}
\label{ssec:optical_trap}

We now consider a system with continuous-valued microstates: the possible positions of a point particle in one dimension, $x \in \mathbb{R}$. Colloidal particles in externally controlled optical traps provide a standard experimental platform for stochastic-thermodynamic protocols~\cite{Blickle_2011}. Feedback-controlled optical tweezers have been used to impose programmable virtual potentials, including double-well landscapes~\cite{Kumar_2018}. These techniques motivate the present one-dimensional model, which still unrestricted control over intermediary Hamiltonians. We show in Appendix~\ref{ap:cont} that the optimal intermediary distributions satisfy the continuum counterpart of~\eqref{eq:lambertsol}.

Given an initial distribution $\rho_0(x)=e^{-\frac{\beta}{2}x^2}/\sqrt{2\pi/\beta}$, which is nonequilibrium with respect to initial Hamiltonian $H_0(x) := \frac{1}{2}(x-x_0)^2$ with $x_0>0$, the controller performs an instantaneous quench to $H'(x) = \frac{1}{2}x^2$, under which $\pi^{H'}=\rho_0$. A series of $N$ intermediary quenches is then carried out. In Sec.~\ref{sec:results} we have assumed arbitrary manipulability of the $N$ intermediary Hamiltonians, so the intermediary distributions need not be translated Gaussians. If instead we were restricted to quenches of the form $H(x-x_i)$ for controllable $x_j$ ($j=1,...,N$), i.e., simple translations of an optical trap, the optimal intermediaries would be equally spaced Gaussians at offset $x_j=jx_0/(N+1)$. We compute the dissipation-minimizing solution~\eqref{eq:lambertsol} to obtain the optimal intermediaries, showing how they differ from simple translations.

\begin{figure}
    \centering
    \includegraphics[width=1\linewidth]{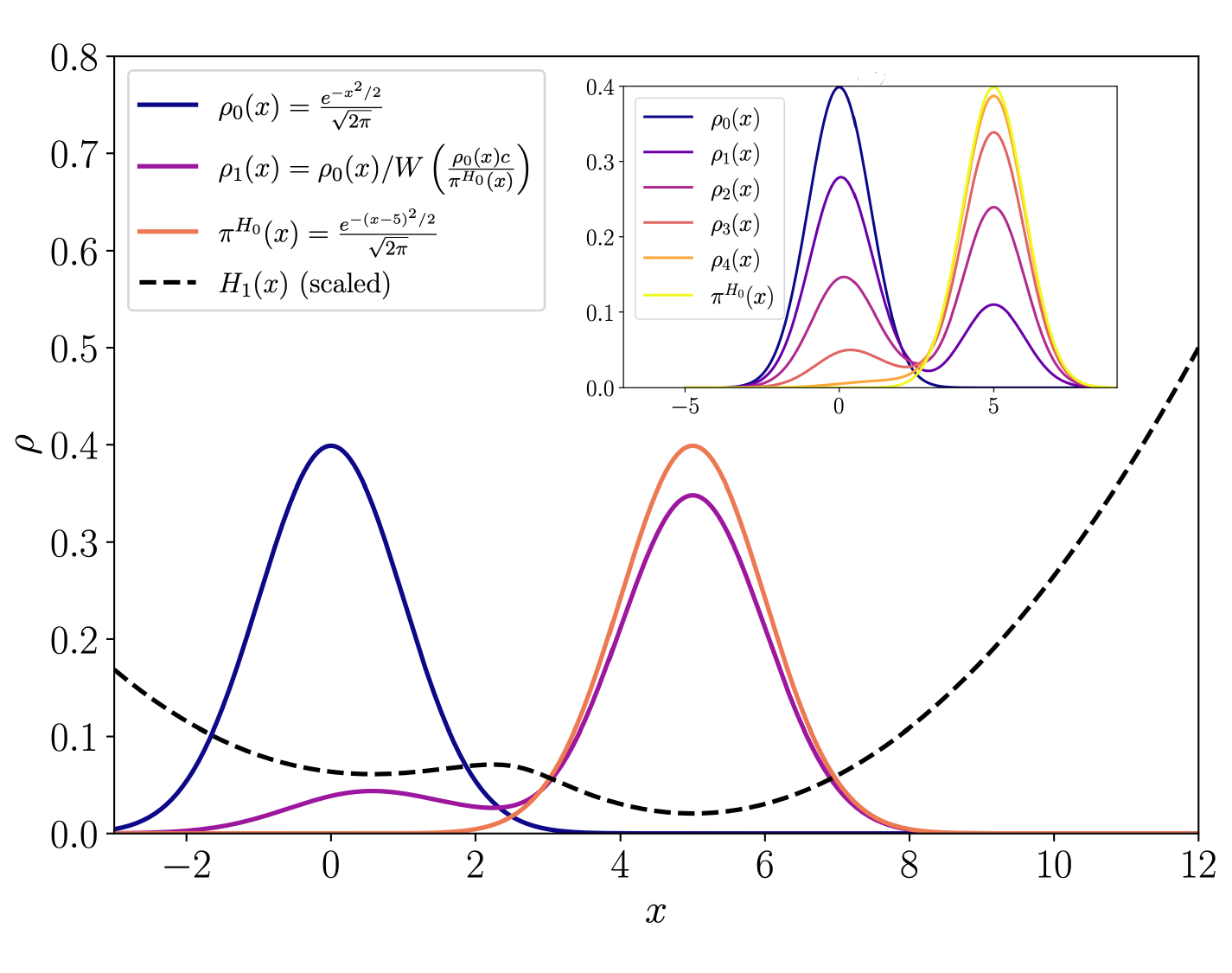}
    \caption{Optimal intermediary $\rho(x)$ for $N=1$ with quadratic Hamiltonians at $\beta=1$, for initial and final distributions $\rho_0(x) = e^{-\frac{1}{2}x^2}/\sqrt{2\pi}$ and $\pi^{H_0}(x)=e^{-\frac{1}{2}(x-5)^2}/\sqrt{2\pi}$. The optimal intermediary Hamiltonian $H_1$ is shown. Similarly the optimal intermediary distributions for $N=4$ case are shown in the inset plot.}
    \label{fig:distsandhamiltonians}
\end{figure}

In the special case $N=1$, the optimal intermediary distribution is given by~\eqref{eq:lambertsol}, which in this case becomes
\begin{equation}
    \rho_1(x) = \frac{e^{-\frac{\beta}{2}x^2}}{\sqrt{\frac{2\pi}{\beta}}W\left(ce^{\frac{\beta}{2}(x_0^2-2xx_0)}\right)},
\end{equation}
where $c$ is determined by normalization; this corresponds via~\eqref{eq:Hprime} to the optimal intermediary Hamiltonian $H_1$:
\begin{equation}
\begin{aligned}
    H_1(x)= \frac{1}{2}(x-x_0)^2-\frac{1}{\beta}W\left(ce^{\frac{\beta}{2}(x_0^2-2xx_0)}\right)+\mathrm{const},
\end{aligned}
\end{equation}
where we use the identity $\log W(z)=\log(z)-W(z)$ for $x>0$.

Fig.~\ref{fig:distsandhamiltonians} shows these optimal intermediary distribution(s) for $N=1$ and $N=4$, alongside the Hamiltonian $H_1(x)\propto -\beta^{-1}\log \rho_1(x)$ for the intermediary equilibrium $\rho_1=\pi^{H_1}$; the inset displays the probability densities ($\rho_0,...,\rho_{N+1}=\pi^{H_0}$). Note that whereas the initial and final distributions are unimodal, the optimal intermediaries can be double-well potentials, differing substantially from any simple translation. Moreover, the optimal intermediaries are not symmetric across the midpoint between $\rho_0$ and $\pi^{H_0}$, and are not exact mixtures of the two---reflecting, respectively, the asymmetry of the KL divergence and that the optimal intermediaries do not lie on the mixture interpolation~\cite{Ay_2019} (Fig.~\ref{fig:distsandhamiltonians}).

We show in Appendix~\ref{ap:klproof} that for $N=1$, as the separation parameter $x_0$ between $\rho_0$ and $\pi^{H_0}$ grows, the dissipation on the last step 
\begin{equation}
    D(\rho_1||\pi^{H_0})\rightarrow 1,
\end{equation}
whereas the first step's dissipation $D(\rho_0||\rho_1)$ grows with $x_0$. We provide formal justification of this for a large class of potentials in Appendix~\ref{ap:klproof}; we also provide numerical evidence that a similar behavior holds for large $N$, where the first leg is dominant and the others approach constants (Fig.~\ref{fig:klsplit}).

\subsection{Enzymatic reaction}
\label{ssec:enzyme}

\begin{figure}
    \centering
    \includegraphics[trim = {0 1cm 0 0}, width=0.85\linewidth]{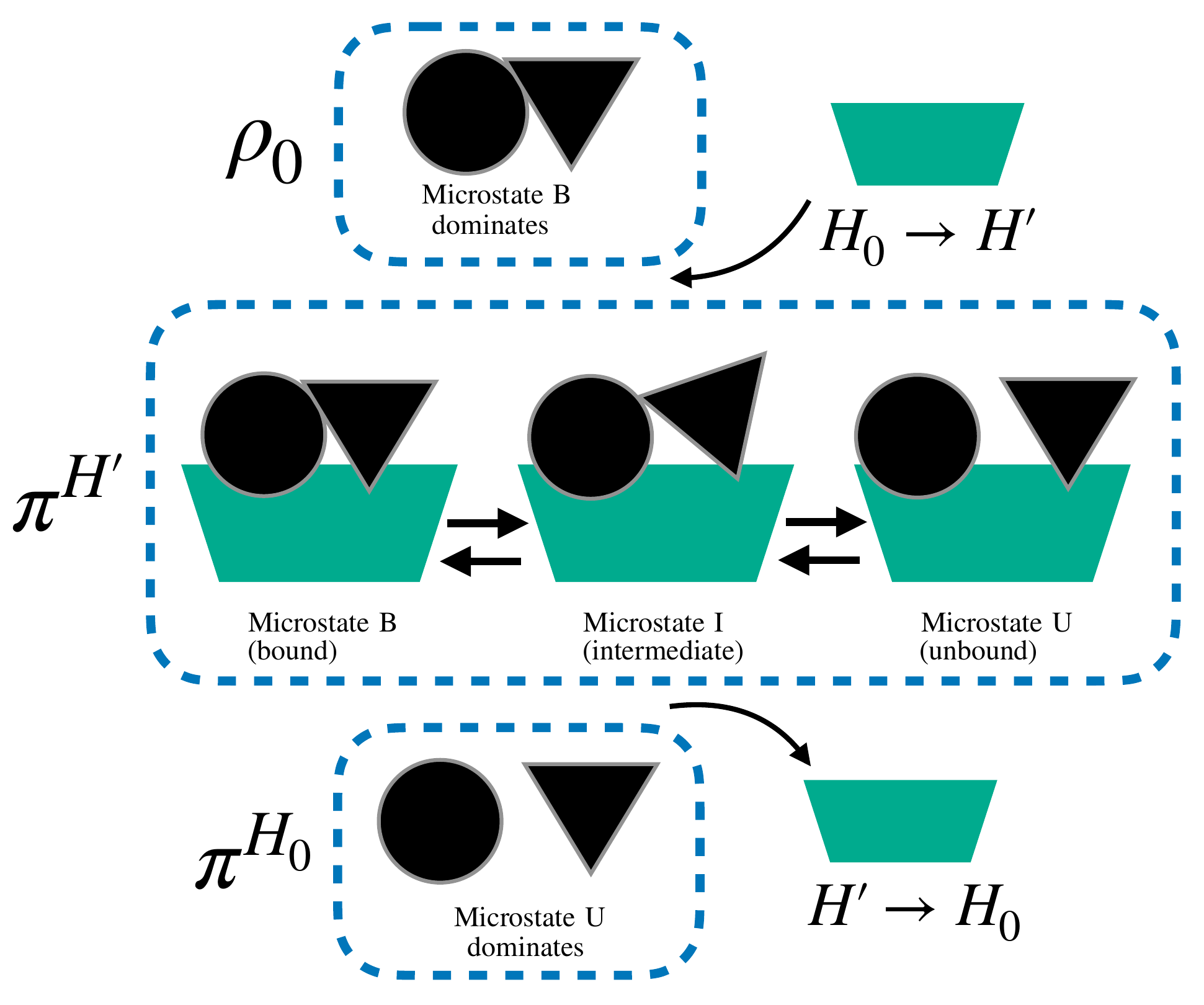}
    \caption{Schemtic of an enzyme-catalyzed reaction as modeled in Sec.~\ref{ssec:enzyme}. Before the enzyme arrives, the substrate is in a meta-stable distribution, $\rho_0$, dominated by the microstate $B$ (bound) representing the substrate before reaction. The enzyme arrives and attaches by van der Waals forces, corresponding to an effective quench $H_0\rightarrow H'$. In the presence of the enzyme, the system relaxes to the equilibrium associated with $H'$. The enzyme eventually disconnects due to thermal fluctuations, corresponding to a second effective quench, $H'\rightarrow H_0$. The system finally relaxes to the equilibrium of the original Hamiltonian, $\pi^{H_0}$, which is dominated by microstate $U$ (unbound).}
    \label{fig:enzyme_schematic}
\end{figure}

We introduce a simple model of enzymatic catalysis and analyze it within the quench-relax framework developed above; an enzyme-induced barrier-height modification serves as an effective quench. The process of interest is catalysis of a dimer's fragmentation into two monomers. The system microstates are $X=\{B,I,U\}$, with $B$ and $U$ representing a bound, high-energy configuration (of energy $\epsilon_B=1$) and an unbound, low-energy configuration (of energy $\epsilon_U=0$), and $I$ represents an intermediate, partially-separated configuration of higher energy than $B$ ($\epsilon_I>1$). See Fig.~\ref{fig:enzyme_schematic}. A similar model has been considered in Ref.~\cite{ArangoRestrepo_2018} in which the sample space was a continuous interval representing a reaction coordinate, with the three microstates of $X$ obtained by coarse-graining. Direct transitions from $B$ to $U$ are forbidden, requiring instead passage through $I$; this additional restriction is not determined by the microstate energies.

Upon the endogenous arrival of an enzyme, an interaction Hamiltonian becomes non-negligible, resulting in an effective adjustment of the microstate energies akin to an externally imposed quench; in idealized conditions, the only nonnegligible effect it has is $\epsilon_I\rightarrow\epsilon_{I}'=\epsilon_I-\delta$, with $\delta>0$ for barrier-lowering and $\delta<0$ for barrier-raising. The composite system has Hamiltonian $H(x,y)=\epsilon_x+H_{\mathrm{int}}(x,y)$, where variable $y$ indicates the enzyme's presence ($y=1$) or absence ($y=0$). We take
\begin{equation}
\label{eq:Hint}
H_{\mathrm{int}}(x,y)=-\delta_{x,I}\delta_{y,1}\delta,
\end{equation}
reflecting that the enzyme's presence only alters the energy of microstate $I$. Assuming slow dynamics of $y$, the statistics of $x$ equilibrate to the conditional $\pi_H(x|y)$, so the effective quench $H_0\rightarrow H'$ is from
\begin{equation}
H_0(x):=H(x,0)=\epsilon_x,
\end{equation}
to
\begin{equation}
\label{eq:H0_Hprime}
H'(x):=H(x,1)=\epsilon_x-\delta_{x,I}\delta.
\end{equation}
Prior to the quench, we assume a nonequilibrium initial condition $\rho_0=(\delta_{x,B})_{x\in\{B,I,U\}}$, representing high concentration of the bound configuration. When the enzyme departs, long after the system's equilibration to $\pi^{H'}$, the reverse effective quench $H'\rightarrow H_0$ takes place; finally, the system equilibrates to $\pi^{H_0}$. Note two key differences from the protocol of Sec.~\ref{sec:results}: first, we operate within a 1D parameterized family ($\delta\in\mathbb{R}$) rather than granting full Hamiltonian manipulability; second, we employ only an initial quench $H\rightarrow H'$ with no subsequent intermediary sequence $\{H_1,...,H_N\}$, i.e., we take $N=0$. Moreover, $H'$ is not the conventional $-\beta^{-1}\log\rho_0$, which indeed is not optimal for finite protocols (Appendix~\ref{ap:initialquench}).

Our results in Appendix~\ref{ap:initialquench} imply that when the initial quench $H'$ is adjustable rather than fixed by~\eqref{eq:Hprime}, the optimal $H'$ at $N=0$ is identical to the optimal $H_1$ at $N=1$ under the fixed initial quench, yielding a $\pi^{H'}$ of the form~\eqref{eq:N1_solution}. However, since we operate within a 1D family of possible $H'$ parameterized by $\delta$ (Eq.~\eqref{eq:H0_Hprime}), we instead optimize over $\delta$ the total dissipated work:
\begin{equation}
\begin{aligned}
\label{eq:Wdiss_enz}
\mathcal{W}^{\mathrm{diss}}(\delta)&=D(\rho_0||\pi^{H'})+D(\pi^{H'}||\pi^{H_0})\\
&=1+\frac{1}{\beta}\log Z_0+\frac{\delta e^{-\beta(\epsilon_I-\delta)}}{Z'(\delta)},
\end{aligned}
\end{equation}
where $Z_0=1+e^{-\beta}+e^{-\beta\epsilon_I}$, $Z'(\delta)=Z_0+e^{-\beta\epsilon_I}(e^{\beta\delta}-1)$, $\pi^{H'}_x=e^{-\beta(\epsilon_x-\delta_{x,I}\delta)}/Z'(\delta)$, and, recall, we set $\epsilon_B=1$ and $\epsilon_U=0$. To determine the dissipation-minimizing value of $\delta$, we differentiate~\eqref{eq:Wdiss_enz} to obtain stationarity condition
\begin{align}
\label{eq:piH2_2}
1+(1-\pi^{H'}_I)\beta\delta=0,
\end{align}
from which $1/\delta^*=-\beta(1-\pi^{H'}_I)<0$, with negativity indicating that the energy barrier should be {\it raised} to minimize dissipation. The solution to~\eqref{eq:piH2_2} is
\begin{equation}
\label{eq:deltastar}
  \delta^*= -\frac{1}{\beta}\left[1+W\left(\frac{e^{\beta(1-\epsilon_I)-1}}{1+e^{\beta}}\right)\right],
\end{equation}
with $W$ denoting the Lambert $W$ function. Enzymatic catalysis accelerates reactions by reducing their activation energy relative to the uncatalyzed process~\cite{robinson_2015}, which would correspond to $\delta>0$. In contrast, the dissipation-minimizing solution $\delta^*<0$ of~\eqref{eq:deltastar} suggests that dissipation minimization omits the benefit of accelerated reaction rate, motivating the model extension introduced below. To quantify the impact of $\delta$ on timescale, we model the dynamics via a continuous-time Markov chain (CTMC) and examine the relaxation time~\cite{levin2017markov} (Appendix~\ref{ap:relaxation_time}). Denoted $\mathcal{T}^{\mathrm{equi}}$, the relaxation time is defined as the reciprocal magnitude of the smallest nonzero eigenvalue of the associated rate matrix; the rates of direct transition between the bound ($B$) and unbound ($U$) microstates are set to zero. Using Arrhenius rates~\cite{Arrhenius1889}, crossing from $B$ to the high-energy intermediate $I$ is far slower than the subsequent jump from $I$ to $U$. Consequently, $\mathcal{T}^{\mathrm{equi}}\approx 1/k_{IB}$, with $k_{IB}$ the rate of transition to $I$ from $B$; up to an undetermined coefficient $M(\beta)$, the result is
\begin{equation}
\label{eq:Tequi}
    \mathcal{T}^{\mathrm{equi}}\approx M(\beta)e^{\beta (\epsilon_I-\delta-\epsilon_B)}.
\end{equation}
Eq.~\eqref{eq:Tequi} illustrates the impact of $\delta^*<0$: the equilibration timescale $\mathcal{T}^{\mathrm{equi}}$ grows exponentially with $-\delta^*$. Thus, the value that minimizes $\mathcal{W}^{\mathrm{diss}}$ does so at the consequence of large $\mathcal{T}^{\mathrm{equi}}$, predicting reaction inhibition \cite{robinson_2015}. This motivates an extension of the optimization problem to also incorporate $\mathcal{T}^{\mathrm{equi}}$.

\subsubsection{Composite cost function $C(\delta)$}
\label{sssec:cost}

Lowering of dissipation can be beneficial~\cite{Pettersson1992,Bruice2002,carter2020escapement}, but the primary functional role of enzymatic catalysis is to increase reaction pace~\cite{arcus2020}. We therefore extend the optimization problem in Sec.~\ref{ssec:enzyme} to incorporate $\mathcal{T}^{\mathrm{equi}}$ through a composite cost function $C(\delta)$ modeling a balance between energetic cost and reaction speed:
\begin{equation}
\label{eq:Cdelta_WT} C(\delta)= \mathcal{W}^{\mathrm{diss}}+\lambda \mathcal{T}^{\mathrm{equi}}.
\end{equation}
For larger $\lambda$, the optimization places greater weight on increased reaction pace; for smaller $\lambda$, it places greater weight lowered dissipation. Under the assumption that an observed enzyme has optimized~\eqref{eq:Cdelta_WT}, the value of $\lambda$ can be estimated from experimental data (Appendix~\ref{ap:estimation}).

The optimum $\delta^*$ satisfies $dC(\delta)/d\delta=0$, which, under the cost function~\eqref{eq:Cdelta_WT}, becomes
\begin{equation}
\label{optdel}
\pi_I^{H'}(1+\beta\delta[1-\pi^{H'}_I])=\lambda'e^{\beta(\epsilon_I-\delta)},
\end{equation}
where $\lambda':=\beta\lambda M(\beta)e^{-\beta}>0$. The solution $\delta^*$ of~\eqref{optdel} is a function of $(\epsilon_I, \lambda',\beta)$, but the $\beta$-dependence is obscured by the undetermined $M(\beta)$, so hereafter we take $\beta$ as constant and evaluate the dependence of $\delta^*$ on $(\epsilon_I,\lambda')$. Specifically, we show that $\delta^*>0$ for sufficiently large $\lambda'$, and thus barrier-lowering is predicted when the optimization is weighted sufficiently towards time-minimization; see Fig.~\ref{fig:deltastar}. To find the threshold $\lambda'$-value for barrier-lowering, we set $\delta=0$ in~\eqref{optdel} from which
\begin{equation}
    \lambda'^*=\pi^{H_0}_Ie^{-\beta\epsilon_I}=\frac{e^{-2\beta\epsilon_I}}{Z_0}=\frac{e^{-2\beta\epsilon_I}}{1+e^{-\beta}+e^{-\beta\epsilon_I}},
\end{equation}
ranging from $\lambda'^*(1)=e^{-2\beta}/(1+2e^{-\beta})$ at the minimal barrier height $\epsilon_I=1$ to $\lambda'^*(\infty)=0$ in the limit $\epsilon_I\rightarrow\infty$ (see Fig.~\ref{fig:deltastar}). Thus when the initial barrier height is larger, a smaller weight towards time-minimization is sufficient for the optimum to favor barrier-lowering.

\begin{figure}
    \centering
\includegraphics[width=\linewidth,trim=80 0 120 0,clip]{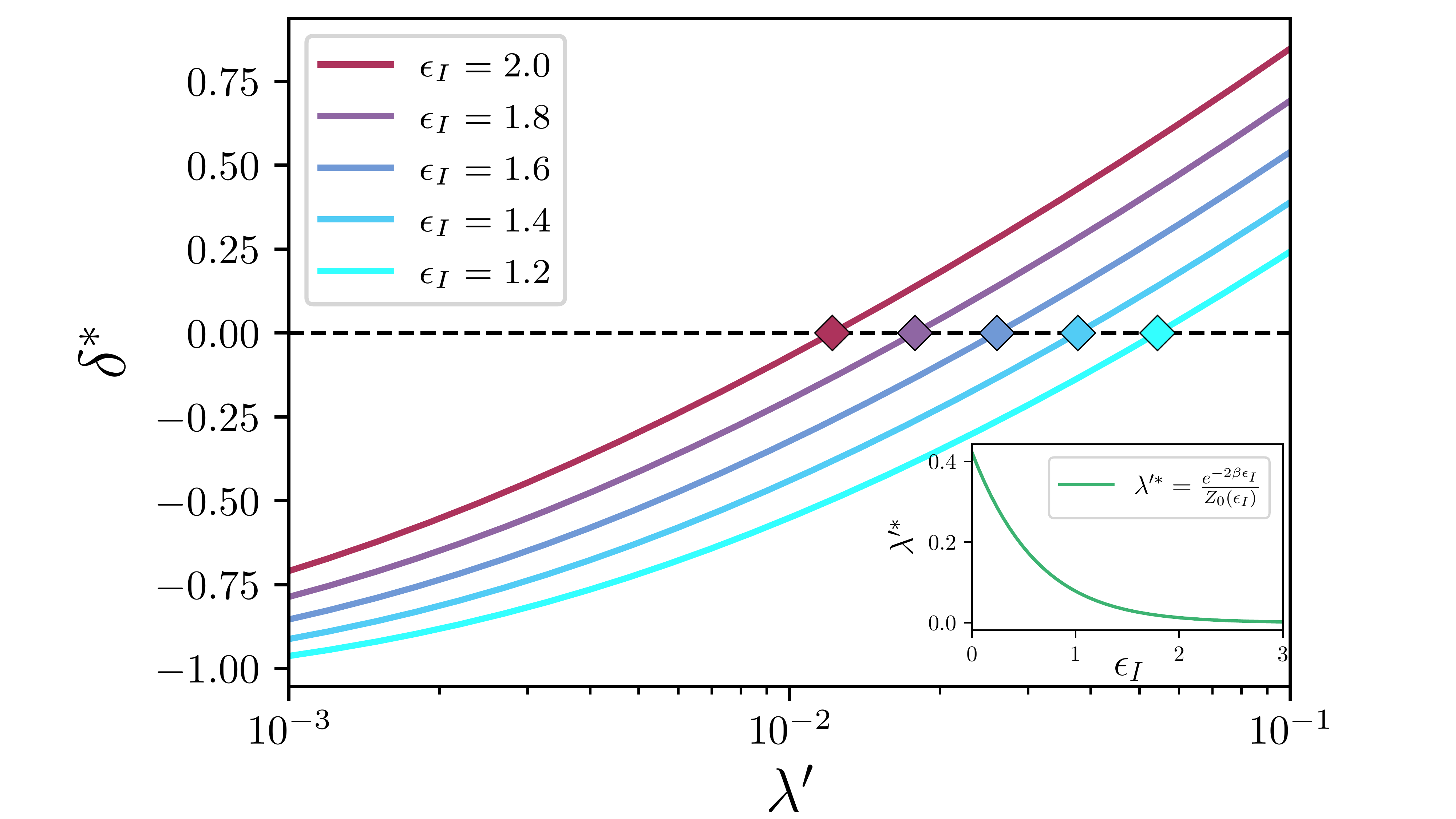}
    \caption{Optimal change in barrier height $\epsilon_I\rightarrow\epsilon_I-\delta$ as a function of the weight towards time-minimization $\lambda'$ in the composite cost function $C(\delta)$ at $\beta=1$, $\epsilon_B=1$, $\epsilon_U=0$. For $\lambda'>\lambda'^*$, the solution $\delta^*$ is positive; diamonds display $(\lambda',\delta)=(\lambda'^*,0)$ at each value of $\epsilon_I$ considered. Inset: decay of $\lambda'^*$ as a function of barrier height $\epsilon_I$; for larger initial barriers, less weight towards time-minimization is required to achieve barrier-lowering.}
    \label{fig:deltastar}
\end{figure}

\section{Discussion}
\label{sec:discussion}

We have studied work extraction from nonequilibrium systems under a quenching protocol for which the relaxation stage is restricted to a finite number $N$ of intermediary Hamiltonians. We obtained the optimal $N$ intermediaries as determined by a coupled system of nonlinear equations involving the Lambert $W$ function. For a single intermediary, $N=1$, we obtained explicit formulae and bounds for mean work extracted. For large $N$ at fixed $|X|$, the mean dissipated work saturates the $O(N^{-1})$ Fisher-Rao lower bound~\cite{Nulton1985}. In the large $N$ limit these optimal intermediaries converge to the Fisher-Rao geodesic. The mean extracted work is $O(1)$ with an $O(N^{-1})$ dissipation correction; we obtained a lower bound on work extraction in terms of Fisher-Rao distance up to error $O(N^{-2})$, applicable to any pair of distributions ($\rho_0,\pi^{H_0})$ in the simplex interior. We considered two extensions of the framework in application to simple physical models: (i) a particle in a tunable potential well, obtaining the optimal intermediary potential, proving that the dominant contribution to dissipation comes from the first step of the protocol in the regime of large separation, and (ii) a minimal model of enzymatic catalysis, showing that minimizing dissipation alone favors reaction inhibition rather than catalysis; incorporating the relaxation time into the cost function produces an optimization with interpretable tradeoff between energetic efficiency and reaction rate, yielding barrier-lowering at sufficiently large weight towards time-minimization; more realistic model extensions may be fit to data. 

This work takes a step towards a broader understanding of nonequilibrium thermodynamics subject to realistic protocol restrictions, particularly with regard to the relaxation phase of quench-relax protocols. We highlight some possible extensions and adjustments that would be worth pursuing in future works. For instance, considering a fixed, predetermined set of intermediaries $\{H_n\}_{n=1}^N$ to choose from, the optimization would become a directed path-finding problem to minimize total dissipation. Another restriction of interest would be to permit control over $N$ intermediaries but restricted to a parameterized family of Hamiltonians, as with the enzyme model in which the change in barrier height $\delta$ is the free parameter of the quench. For the optimization problem arising in the enzyme model, another realistic constraint could be to allow $\mathcal{T}^{\mathrm{equi}}$ to equal {\it at most} a certain value, rather than directly including it in the cost function, or to incorporate other quantifiers of timescale such as mean first-passage time, which may be more relevant to biological processes than relaxation time. There are many related thermodynamic optimization problems of interest involving both timescale and dissipation; for instance, the extension of the enzyme cost function to an $N$-step intermediary sequence. These, among many related possibilities, would help further bridge the gap between idealized thermodynamic models and realistic work extraction scenarios.

More broadly, many nonequilibrium work-extraction problems in physical, engineered, and biological systems operate under restricted control rather than the idealized assumption of arbitrary Hamiltonian manipulability~\cite{kolchinsky2021work}. Understanding optimal protocols under such restrictions may therefore be useful both for designing work extraction schemes and for formulating more realistic models of biological processes such as free energy transduction~\cite{brown2019theory,hill1983some,hill2012free,leighton2025flow}. The finite relaxation problem studied here provides a setting in which the thermodynamic cost of limited control can be tractably characterized. Future works involving more realistic models can explore connections to data, for instance, measurements of the barrier-lowering in real enzymes. More generally, we hope that future works will further the investigation of optimization principles subject to constrained action protocols and their implications for the function of artificial and natural systems.

\bibliography{references}

\appendix

\section{Unrestricted initial quench}
\label{ap:initialquench}

In the main text, we have considered finite relaxation protocols for which the first quench $H_0\rightarrow H'$ is always to a Hamiltonian $H'$ of the form
\begin{equation}
\label{eq:Hprime_logrho0}
    H'(x):=-\beta^{-1}\log\rho_0(x),
\end{equation}
so that the resulting equilibrium distribution is $\rho'=\rho_0$ (the initial nonequilibrium with respect to $H_0$). For quasistatic relaxation protocols the choice~\eqref{eq:Hprime_logrho0} is optimal, enabling complete harvesting of the initial excess free energy $\Delta F_{H_0}(\rho_0)$ \cite{parrondo2015thermodynamics}. We have considered protocols for which instead of an infinitely gradual return from $H'$ to $H_0$, a set of $N$ intermediaries $(H_1,...,H_N)$ is utilized. 

However, it is natural to ask whether the choice of $H'$ in~\eqref{eq:Hprime_logrho0} remains optimal in finite relaxation scenarios, or if it is only optimal in the quasistatic limit $N\rightarrow\infty$. Indeed, we now show that a different initial quench $H'$ is optimal in the extended problem of choosing $N+1$ quenches $(H',H_1,...,H_N)$ to minimize $\mathcal{W}^{\mathrm{diss}}$. Below, we characterize this extended optimization, demonstrating that its solution can be expressed via a mild adjustment to the results presented in Sec.~\ref{sec:results}.

To determine the optimal protocol when the initial quench $H_0\rightarrow H'$ is adjustable rather than set by~\eqref{eq:Hprime_logrho0}, we do an accounting of the total work extraction accumulated from each step. For two Hamiltonians $H$ and $H^*$ and a distribution $\rho$, we denote the work extracted via the quench $H\rightarrow H^*$ as
\begin{equation}
\mathcal{W}^{\mathrm{ext}}_{H\rightarrow H^*}(\rho)=-\sum_{x\in X}(H^*(x)-H(x))\rho(x).
\end{equation}
For the sequence of quenches $(H',H_1,...,H_N)$, the total work extraction of the protocol is
\begin{equation}
\label{eq:total_Wext_Hprime}
\mathcal{W}^{\mathrm{ext}}=\mathcal{W}^{\mathrm{ext}}_{H_0\rightarrow H'}(\rho_0)+\mathcal{W}^{\mathrm{ext}}_{H'\rightarrow H_1}(\rho')+\sum_{i=1}^N\mathcal{W}^{\mathrm{ext}}_{H_i\rightarrow H_{i+1}}(\rho_i),
\end{equation}
with $\rho'=\pi^{H'}$, $\rho_{i}=\pi^{H_i}$ for $i\in[N]$, and $\rho_{N+1}=\pi^{H_0}$. The work extracted in the initial quench $H_0\rightarrow H'$ satisfies
\begin{equation}
\label{eq:Wext_H0_Hprime}
\beta\mathcal{W}^{\mathrm{ext}}_{H_0\rightarrow H'}(\rho_0)= D(\rho_0||\pi^{H_0})-D(\rho_0||\rho')+\log\frac{Z'}{Z_0}.
\end{equation}
In the next quench $H'\rightarrow H_1$,
\begin{equation}
\label{eq:Wext_Hprime_H1}
\beta\mathcal{W}^{\mathrm{ext}}_{H'\rightarrow H_1}(\rho')=-D(\rho'||\pi^{H_1})+\log\frac{Z_1}{Z'}.
\end{equation}
For each subsequent quench $H_{i}\rightarrow H_{i+1}$,
\begin{equation}
\label{eq:Wext_Hi_Hiplus1}
\beta\mathcal{W}^{\mathrm{ext}}_{H_i\rightarrow H_{i+1}}(\pi^{H_i})= -D(\pi^{H_i}||\pi^{H_{i+1}})+\log\frac{Z_{i+1}}{Z_i}.
\end{equation}
The total work extraction~\eqref{eq:total_Wext_Hprime} is then obtainable by summing~\eqref{eq:Wext_H0_Hprime},~\eqref{eq:Wext_Hprime_H1}, and~\eqref{eq:Wext_Hi_Hiplus1} for each $i\in[N]$. The log-partition functions cancel, and noting the appearance of $\Delta F_{H_0}(\rho_0):=\beta^{-1}D(\rho_0||\pi^{H_0})$ in~\eqref{eq:Wext_H0_Hprime}, the total {\it dissipated} work $\mathcal{W}^{\mathrm{diss}}:=\Delta F_{H_0}(\rho_0)-\mathcal{W}^{\mathrm{ext}}$ reduces to a sum of intermediary KL divergences:
\begin{equation}
\label{eq:total_dissipation_rhoprime}
\beta\mathcal{W}^{\mathrm{diss}}=D(\rho_0||\rho')+D(\rho'||\pi^{H_1})+\sum_{i=1}^ND(\pi^{H_i}||\pi^{H_{i+1}}).
\end{equation}
Accordingly, the problem of determining the optimal initial quench alongside the intermediaries, i.e., choosing all of ($H',H_1,...,H_N$), is identical to the original optimization problem of Sec.~\ref{sec:results} at fixed $H'=\beta^{-1}\log\rho_0$, but with the replacement of $N$ by $N+1$. Hence by permitting adjustment of $H'$ as part of the optimization, the maximum extractable work is improved from $\mathcal{W}_N^{\mathrm{ext}}$ to $\mathcal{W}_{N+1}^{\mathrm{ext}}$, where $\mathcal{W}^{\mathrm{ext}}_{N}$ (Eq.~\eqref{eq:Wext_def}) denotes the maximal work extraction in the original scenario (achieved by intermediaries~\eqref{eq:lambertsol}). 

At small $N$, this difference may be substantial. For large $N$, it becomes insignificant, vanishing as $O(N^{-2})$ in contrast with the $O(1)$ work extracted and $O(N^{-1})$ dissipated. Specifically, using~\eqref{eq:FR_asymptotic} (Appendix~\ref{ap:fisherrao}),
\begin{equation}
\begin{aligned}
&\mathcal{W}_{N+1}^{\mathrm{ext}}-\mathcal{W}_N^{\mathrm{ext}}=-\mathcal{W}_{N+1}^{\mathrm{diss}}+\mathcal{W}_N^{\mathrm{diss}}\\
&=-\frac{d_{\mathrm{FR}}(\rho_0,\pi^{H_0})}{2\beta}\left(\frac{1}{N+1}-\frac{1}{N}\right)+O(N^{-2})\\
&=\frac{d_{\mathrm{FR}}(\rho_0,\pi^{H_0})^2}{2\beta N(N+1)}+O(N^{-2})=O(N^{-2}).
\end{aligned}
\end{equation}

\section{Optimal intermediary distributions}
\label{ap:lambertproof}

Here we provide a derivation of~\eqref{eq:lambertsol}, the coupled Lambert equations determining the dissipation-minimizing solution. The problem is identical to the minimization of the `chained KL divergence' of Ref.~\cite{Pavlichin_2016}. In particular, the $k$-fold chained KL divergence between two distributions $(p,q)$ is 
\begin{equation}
D^{(k)}(p||q):=\underset{w_1,...,w_{k-1}}{\min}\sum_{i=1}^kD(w_i||w_{i-1}),
\end{equation} 
with $w_0=q$ and $w_k=p$. In the work extraction protocol of Sec.~\ref{sec:framework}, we seek optimal intermediaries $\{\rho_i\}_{i=1}^N$ for which $D^{(N+1)}(\rho_0||\pi^{H_0})$ is minimized.

\begin{proposition}
\label{prop:1}
    Let $\alpha$ and $\gamma$ be probability distribution on a finite sample space $X$, with $\gamma(x) > 0$ for all $x \in X$. Denoting $\rho_0:=\alpha$ and $\rho_{N+1}:=\gamma$, the optimal intermediary distributions $\{\rho_i\}_{i=1}^N$ minimizing $\sum_{i=1}^{N+1}D(\rho_{i-1}||\rho_{i})$ satisfy
    \begin{equation}
    \label{eq:prop1}
        \rho_i(x) = \frac{\rho_{i-1}(x)}{W(c_i \rho_{i-1}(x) / \rho_{i+1}(x))}, \  \ \ \ i=1,...,N,
    \end{equation}
    where $W$ is the principal Lambert $W$ function and the constants $c_i > 0$ are fixed by normalization conditions $\sum_{x \in X} \rho_i(x) = 1$.
\end{proposition}

\begin{proof}
We impose normalization with Lagrange multipliers. Differentiating $\sum_{i = 1}^{N+1} D(\rho_{i-1} \| \rho_{i})$ with respect to $\rho_i(x)$ gives, for each $i \in [N]$ and $x\in X$, the stationarity condition
\begin{equation}
\label{eq:optimal_intermediaries}
    1+\lambda_i-\frac{\rho_{i-1}(x)}{\rho_i(x)}+\log\frac{\rho_i(x)}{\rho_{i+1}(x)}=0.
\end{equation}

Let $u(x):=\frac{\rho_{i-1}(x)}{\rho_i(x)}$. Rearranging Eq.~\eqref{eq:optimal_intermediaries} gives 
\begin{equation}
    u(x)e^{u(x)} = e^{1+\lambda_i}\frac{\rho_{i-1}(x)}{\rho_{i+1}(x)}. 
\end{equation} 
Thus, writing $c_i=e^{1+\lambda_i}>0$, 
\begin{equation}
    u(x)=W\left( \frac{c_i\rho_{i-1}(x)}{\rho_{i+1}(x)} \right),
\end{equation}
from which~\eqref{eq:optimal_intermediaries} follows by the definition of $u(x)$. The constants $\{c_i\}_{i=1}^N$ are fixed by the $N$ normalization conditions $\sum_x\rho_i(x)=1$.
\end{proof}

When $\rho_{i-1}(x)=0$, Eq.~\eqref{eq:prop1} is understood by continuity, yielding $\rho_i(x)=\rho_{i+1}(x)/c_i$. Note the relation $\lambda_i=-D(\rho_{i}||\rho_{i+1})$, obtained by multiplying~\eqref{eq:optimal_intermediaries} by $\rho_i(x)$ and summing over $x\in X$~\cite{nielsen2013symmetrical}. Moreover, by~\eqref{eq:H_i_opt}, the $(i+1)$st optimal Hamiltonian is obtainable by additive adjustment to the $i$th Hamiltonian,
\begin{equation}
 H_i(x)=H_{i-1}(x)+\beta^{-1}\log W\left(\frac{c_i\rho_{i-1}(x)}{\rho_{i+1}(x)}\right).
\end{equation}

\section{Computing the optimal intermediaries}
\label{ap:algorithm}

Here we describe a nonlinear least-squares implementation used to generate this paper's distributions and Hamiltonians. We have shown that the optimal intermediary distributions of~\eqref{eq:lambertsol} exist and are unique; here, we provide a straightforward algorithm to compute them. The solution we compute is denoted $\{\hat{\rho}_i(x)\}_{i\in[N],x\in X}$. We construct residuals from Eq.~\eqref{eq:lambertsol} and normalization. Namely,
\begin{equation}
\label{eq:rix}
    r_i(x) := \hat\rho_i(x) - \frac{\hat\rho_{i-1}(x)}{W\left(
    \frac{\hat\rho_{i-1}(x)c_i}{\hat\rho_{i+1}(x)}
    \right)}\\
\end{equation}
for each $i\in[N]$ and $x\in X$, and
\begin{equation}
\label{eq:qi}
    q_i : = \sum_{x\in X}\hat{\rho}_i(x)-1,
\end{equation}
for each $i\in [N]$. All $N(M+1)$ of the residuals~\eqref{eq:rix},~\eqref{eq:qi} equal zero for the optimal solution~\eqref{eq:lambertsol}. Accordingly, we can minimize the quantity
\begin{equation}
    R=\sum_{i\in[N]} \left(\sum_{x\in X}r_i(x)^2+q_i^2\right),
\end{equation}
subject to positivity constraints on each probability value. A nonlinear least squares solver in {\it SciPy} \cite{2020SciPy-NMeth} returns a unique solution corresponding to $R=0$ up to machine precision. The solver requires specification of an initial guess for the values being computed; we take the linear interpolation  $\frac{1}{N+1}[(N + 1 - i)\alpha(x) + i\gamma(x)]$ for $\hat{\rho}_i(x)$, and the value $1$ for $c_i$.

\section{Background on Fisher-Rao metric}
\label{ap:fr_background}

The space of probability distributions is equipped with a Riemannian metric given by the Fisher information matrix (Fisher-Rao metric) \cite{Ay_2019,amari2016information}. For a distribution $p_{\boldsymbol{\theta}}(x)$ parameterized by a vector $\boldsymbol{\theta}=(\theta_1,\ldots,\theta_m)$ the Fisher-Rao metric $g=(g_{jk})_{(j,k)\in[m]^2}$ is defined componentwise by
\begin{equation}
\label{eq:gjk}
    g_{jk}(\boldsymbol{\theta}) := \mathbb{E}_{p_{\boldsymbol{\theta}}}\left[ \frac{\partial\log p_{\boldsymbol{\theta}}(x)}{\partial\theta_j}\, \frac{\partial \log p_{\boldsymbol{\theta}}(x)}{\partial\theta_k} \right].
\end{equation}
The Fisher-Rao metric can be obtained by expanding the KL divergence of $p_{\boldsymbol{\theta}}(x)$ to $p_{\boldsymbol{\theta}+\Delta\boldsymbol{\theta}}(x)$ to second order in $\Delta\boldsymbol{\theta}(x)$. Note that this is a local metric and must be evaluated at a point on the simplex. A Fisher-Rao geodesic is a locally length minimizing curve with respect to this metric.

In this work we consider categorical distributions $p=(p(x))_{x\in X}$ over finite sample spaces $X$ of size $M=|X|$. By normalization, the number of independent degrees of freedom is $m=M-1$. We parameterize $p=p_{\boldsymbol{\theta}}$ directly by the first $M-1$ probabilities, taking $\theta_i=p(x_i)$, so that $p_{\boldsymbol{\theta}}(x_i)=p(x_i)$ for $i\in\{1,...,M-1\}$ and $p_{\boldsymbol{\theta}}(x_M)=1-\sum_{i=1}^Mp(x_i)$. For general $i$,
\begin{equation}
\frac{\partial\log p(x_i)}{\partial p(x_j)}=\frac{\delta_{i,j}(1-\delta_{i,M})}{p(x_i)}-\frac{\delta_{i,M}}{p(x_M)}
\end{equation}
from which, for $(j,k)\in[M-1]^2$, the components~\eqref{eq:gjk} of $g(p_{\boldsymbol{\theta}}):=g(\boldsymbol{\theta})$ are given by
\begin{equation}
    g_{jk}(p)=\frac{\delta_{j,k}}{p(x_j)}+\frac{1}{p(x_M)}.
\end{equation}

For two distributions $\alpha$ and $\gamma$, and the set of all continuous paths $C(\alpha,\gamma)$ connecting them, with each element $p\in C(\alpha,\gamma)$ parameterized by the first $M-1$ components of a normalized distribution: $p(t)=(p_1(t),...,p_{M-1}(t))$; also denote $\dot{p}(t):=\frac{d}{dt}p(t)$. Define the local inner product of two vectors $(u,v)\in\mathbb{R}^{M-1}\times\mathbb{R}^{M-1}$ with respect to the metric $g(p)$ by
\begin{equation}
v^Tg(p) u: = \sum_{j=1}^{M-1}\sum_{k=1}^{M-1}v_jg_{jk}(p)u_{k}.
\end{equation}

 The path length of contour $p=(p(t))_{t\in[0,1]}$ is
 \begin{equation}
 \label{eq:path_length}
    d_g[p]:=\int_0^1\sqrt{ \dot{p}(t)^Tg(p(t))\dot{p}(t) dt},
\end{equation}
and the Fisher-Rao distance is defined as
 \begin{equation}
 \label{eq:FR_def}
    d_{\mathrm{FR}}(\alpha, \gamma) := \min_{p\in C(\alpha,\gamma)} d_g[p].
    \end{equation}
    We also define the `action' \cite{jost_riemannian_2017} or `thermodynamic divergence'~\cite{crooks2007measuring} of a path as
    \begin{equation}
    \mathcal{A}_g[p]:=\int_0^1 \dot{p}(t)^T g(p(t))\dot{p}(t)\, dt,
\end{equation}
The path of minimum thermodynamic divergence is the geodesic:
\begin{equation}
p^*_{\alpha,\beta}=(p^*_{\alpha,\beta}(t))_{t\in[0,1]}=\underset{p\in C(\alpha,\gamma)}{\mathrm{argmin}}\mathcal{A}_g[p].
\end{equation}

The Cauchy-Schwarz inequality in the form $\int_0^1 f(t)dt\le[\int_0^1 f(t)^2dt]^{\frac{1}{2}}$ can be applied, in this case
$\int_0^1 [\dot{p}^Tg(p)p]^{\frac{1}{2}}\, dt\le[\int_0^1\dot{p}^Tg(p)p \, dt]^{\frac{1}{2}}$, or equivalently $\mathcal{A}_g[p]\ge d_{g}[p]^2$ for any path $p\in C(\alpha,\gamma)$. Moreover, equality is achieved only on the geodesic:
\begin{equation}
\sqrt{\mathcal{A}_g[p^*_{\alpha,\gamma}]}=d_{g}[p^*_{\alpha,\gamma}]=d_{\mathrm{FR}}(\alpha,\gamma).
\end{equation}

 For two categorical distributions $p(x)$ and $q(x)$, consider the change of variables $\mu_i=\sqrt{p(x_i)}$, $\nu_i=\sqrt{q(x_i)}$. Then $\mu=(\mu_i)_{i=1}^M$ and $\nu=(\nu_i)_{i=1}^M$ reside in the positive orthant of the unit $(M-1)$-sphere: $\sum_{i=1}^M\mu_i^2=\sum_{i=1}^M\nu_i^2=1$. The great circle connecting $\mu$ and $\nu$ corresponds to the Fisher-Rao geodesic~\cite{amari2016information}. The Fisher-Rao distance~\eqref{eq:FR_def} between $p$ and $q$ has the value~\cite{miyamoto_fr_distance} 
\begin{equation}
\label{eq:FR_arccos}
d_{\mathrm{FR}}(p,q)=2\mathrm{arccos}(\mu^T \nu),
\end{equation}
with $\mu^T\nu=\sum_{i=1}^M\mu_i\nu_i=\sum_{i=1}^M\sqrt{p(x_i)q(x_i)}$.

\section{Leading-order minimal dissipated work}
\label{ap:fisherrao}

The minimal dissipated work $\mathcal{W}_N^{\mathrm{diss}}$ in the $N$-step finite relaxation protocol of Sec.~\ref{sec:framework} is equal to $D^{(N+1)}(\rho_0||\pi^{H_0})$, a `chained KL divergence' in the terminology of Ref.~\cite{Pavlichin_2016}. The $N\rightarrow\infty$ asymptotic of $\mathcal{W}_{N}^{\mathrm{diss}}$ has long been established for step-equilibration processes~\cite{Nulton1985,Diosi2000}; namely, Eq.~\eqref{eq:FR_asymptotic}. This follows from Proposition~\ref{prop:2} below, and can also be deduced from the results of Ref.~\cite{Pavlichin_2016}.

\begin{proposition}
\label{prop:2}
Let $\mathcal{W}^{\mathrm{diss}}_N$ be defined as in~\eqref{eq:WdissN} and denote $\rho_0=\alpha$ and $\pi^{H_0}=\gamma$. Assume $\alpha$ and $\gamma$ have the same support. Then
\begin{equation}
\label{eq:FR_asymptotic_app}
    \mathcal{W}^{\mathrm{diss}}_N
    =
    \frac{d_{\mathrm{FR}}(\alpha,\gamma)^2}
    {2\beta N}
    +
    O(N^{-2}),
\end{equation}
as $N\rightarrow\infty$, where $d_{\mathrm{FR}}(\alpha,\gamma)$ is the Fisher-Rao distance between $\alpha$ and $\gamma$.
\end{proposition}

\begin{proof}
    Consider a the unit simplex over finite sample space $X$, and a smooth parameterized path $p\in C(\alpha,\gamma)$ specified by its first $M-1$ components $p_1(t),...,p_{M-1}(t)$, with argument $t\in [0, 1]$ and endpoints $p_i(0) = \alpha(x_i)$ and $p_i(1) = \gamma(x_i)$ for $i=1,...,M-1$. We consider $N + 2$ equally spaced values of $t$, denoted $t_i = i/(N+1)\in[0,1]$ for $i=0,...,N+1$ ($t_0=0$, $t_{N+1}=1$) and denote $\Delta t=t_{i+1}-t_i=\frac{1}{N+1}=O(N^{-1})$.

 Denoting $g(p)$ for the Fisher-Rao metric on the unit simplex at $p$~\cite{amari2016information} with components given by~\eqref{eq:gjk}, for each pair $(p(t_i),p(t_{i+1}))$, the KL divergence can be expanded as
    \begin{equation}
        D(p(t_i) \| p(t_i + \Delta t)) = 
        \frac{1}{2} \dot{p}(t_i)^Tg(p(t_i))\dot{p}(t_i) + O(\Delta t^3).
    \end{equation}
   
Writing $D_N:=\sum_{i = 0}^{N}D(p(t_i) \| p(t_{i+1}))$, we obtain
    \begin{equation}
        D_N= \frac{1}{2} \sum_{i=0}^N\dot{p}(t_i)^Tg(p(t_i))\dot{p}(t_i)\Delta t^2 + O(\Delta t^2).
\end{equation}

    Since $\dot{p}^T(t) g(p(t)) \dot{p}(t)$ is continuous, the Riemann sum can be replaced with the integral:
    \begin{equation}
        \sum_{i = 0}^N \dot{p}(t_i)^Tg(p(t_i))\dot{p}(t_i)\Delta t = \int_{0}^1 \dot{p}(t)^Tg(p(t))\dot{p}(t)dt + O(\Delta t).
    \end{equation}
    
    Therefore,
    \begin{align}
        D_N=
        \frac{1}{2(N+1)}
        \int_0^1
        \dot{p}(t)^Tg(p(t))\dot{p}(t) dt
        +O(N^{-2}).
    \end{align}
  
 Writing $\mathcal{W}_N^{\mathrm{diss}}=\beta^{-1}D_N$ and minimizing $D_N$ over the set $C(\alpha,\gamma)$ of continuous paths $p(t)$ from $\alpha$ to $\gamma$ yields the minimum thermodynamic divergence $\mathcal{A}(\alpha,\gamma)$ between $\alpha$ and $\gamma$:
\begin{equation}
\begin{aligned}
\label{eq:Wdiss_FR_ap}
       \mathcal{W}^{\mathrm{diss}}_N &= \min_{p\in\mathcal{C}(\alpha,\gamma)}\frac{\mathcal{A}(\alpha,\gamma)}{2\beta (N + 1)} + O(N^{-2})\\
       = &\frac{\mathcal{A}_g[p^*_{\alpha,\beta}]}{2\beta (N + 1)} + O(N^{-2})= \frac{d_{\mathrm{FR}}(\alpha,\gamma)^2}{2\beta (N + 1)} + O(N^{-2})
   \end{aligned} 
   \end{equation}
having used that the geodesic saturates the bound $\mathcal{A}_g[p]\ge d_g[p]^2$ (Appendix.~\ref{ap:fr_background}). Using $\frac{1}{N+1}=\frac{1}{N}+O(N^{-2})$ in \ref{eq:Wdiss_FR_ap} we arrive at~\eqref{eq:FR_asymptotic_app}.

\end{proof}

\section{Lower bound on optimal work extraction}
\label{ap:Wext_lower}

To obtain the bound~\eqref{eq:Wext_lower_FR}, we first write the extracted work with optimal intermediaries (Eq.~\eqref{eq:Wext_def}) using~\eqref{eq:FR_asymptotic}, as
\begin{equation}
   \mathcal{W}^{\mathrm{ext}}_N
        =
        \beta^{-1}D(\rho_0\Vert\pi^{H_0})
        -
        \frac{d_{\mathrm{FR}}(\rho_0,\pi^{H_0})^2}
        {2\beta N}
        +
        O(N^{-2}),
\end{equation}
which we lower bound by making use of the following Lemma.

\begin{lem}
\label{lem:KL_lower_FR}
For any two categorical distributions $(p,q)$ over the finite sample space $X$, it holds that
\begin{equation}
\label{eq:lemma_KL_lower_FR}
D(p||q)\geq -2\log \cos\frac{d_\mathrm{FR}(p,q)}{2},
\end{equation}
and in particular,
\begin{equation}
D(p||q)\ge \frac{d_{\mathrm{FR}}(p,q)^2}{4}.
\end{equation}
\end{lem}

\noindent {\it Proof.} We let $M:=|X|$ and write the components of $(p,q)$ as $(p(x),q(x))$ for $x\in X$. The Bhattacharaya distance $d_B(p,q)$~\cite{bhattacharyya1946measure} is defined as
\begin{align}
\label{eq:Bhattacharaya}
d_\mathrm{B}(p,q) :=-\log\sum_{x\in X}\sqrt{p(x) q(x)} .
\end{align}
Note that
\begin{equation}
\begin{aligned}
\label{eq:KL_B}
    D(p||q)&=\sum_{x\in X} p(x)\log(p(x)/q(x)) = -2 \sum_i p_i\log\frac{\sqrt{q(x)}}{\sqrt{p(x)}}\\
    &\geq -2\log\sum_{x\in X} p(x)\frac{\sqrt{q(x)}}{\sqrt{p(x)}}= 2d_\mathrm{B}(p,q),
\end{aligned}
\end{equation}
where we used Jensen's inequality. Recall~\eqref{eq:FR_arccos} for the Fisher-Rao distance between two categorical distributions, from which
\begin{equation}
\label{eq:cos}
    \sum_{i=1}^{M}\sqrt{p(x) q(x)}=\cos\frac{d_\mathrm{FR}(p,q)}{2}.
\end{equation}
Using~\eqref{eq:cos} in~\eqref{eq:Bhattacharaya}, we have 
\begin{equation}
    d_\mathrm{B}(p,q) =-\log \cos\frac{d_\mathrm{FR}(p,q)}{2},
\end{equation}
which, when applied in~\eqref{eq:KL_B}, yields Eq.~\eqref{eq:lemma_KL_lower_FR},
\begin{equation}
\label{eq:KL_ineq_logcos}
    D(p||q)\geq -2\log \cos\frac{d_\mathrm{FR}(p,q)}{2}.
\end{equation}
Using $\cos y\leq e^{-y^2/2}$ for $0\leq y<\pi/2$, and hence $-2\log(\cos y)\geq y^2$,
\begin{equation}
    D(p||q)\geq \frac{d_\mathrm{FR}(p,q)^2}{4}.
\end{equation}
$\square$

\ \\
From Lemma~\ref{lem:KL_lower_FR} with $p=\rho_0$ and $q=\pi^{H_0}$, the work extraction bound~\eqref{eq:Wext_lower_FR} presented in Sec.~\ref{sec:results} follows:
\begin{align}
        \mathcal{W}^{\mathrm{ext}}_N
        &=
        \beta^{-1}D(\rho_0\Vert\pi^{H_0})
        -
        \frac{d_{\mathrm{FR}}(\rho_0,\pi^{H_0})^2}
        {2\beta N}
        +
        O(N^{-2})\\
        &\gtrsim \beta^{-1}d_{\mathrm{FR}}(\rho_0,\pi^{H_0})^2\left(\frac{1}{4}-\frac{1}{2N}\right).
\end{align}
Moreover, a stronger bound is obtained via the first inequality of Lemma~\ref{lem:KL_lower_FR}, the penultimate~\eqref{eq:KL_ineq_logcos} above. Namely,
\begin{align}
\beta\mathcal{W}^{\mathrm{ext}}_N
        \gtrsim -2\log\cos\frac{d_{\mathrm{FR}}(\rho_0,\pi^{H_0})}{2}-\frac{d_{\mathrm{FR}}(\rho_0,\pi^{H_0})^2}{2N}.
\end{align}

\section{Properties of $N=1$ optimal intermediary}
\label{ap:propsoflambert}

Given a finite set of microstates $X$ and initial and final distributions $\alpha(x)$ and $\gamma(x)$, the optimal intermediary for $N=1$ is given by~\eqref{eq:N1_solution}, namely,
\begin{equation}
\label{eq:N1_intermed}
    \rho(x) = \frac{\alpha(x)}{W(c\alpha(x)/\gamma(x))},
\end{equation}
where $c$ is determined by normalization of $\rho$. The same Lambert-$W$ form includes the Jeffreys-centroid problem for the corresponding arithmetic- and geometric-mean endpoints~\cite{nielsen2013symmetrical}; arbitrary endpoint pairs need not define a Jeffreys centroid. The minimal dissipated work is specified by
\begin{equation}
\begin{aligned}
\label{eq:Wdiss_N1_1}
    \beta \mathcal{W}^\mathrm{diss}_1 &= \sum_{i=1}^{N+1}D(\rho_{i-1} ||\rho_i) = D(\alpha||\rho)+D(\rho||\gamma)\\
    &=\sum_{x\in X}\left(\alpha(x)\log \frac{\alpha(x)}{\rho(x)} + \rho(x)\log \frac{\rho(x)}{\gamma(x)}\right).
\end{aligned}
\end{equation}

Applying~\eqref{eq:N1_intermed} in~\eqref{eq:Wdiss_N1_1} and rearranging,
\begin{equation}
\begin{aligned}
    &\beta \mathcal{W}_1^\mathrm{diss} =\sum_{x\in X}\alpha(x)\left( 1 + \log W\left(\frac{c\alpha(x)}{\gamma(x)}\right) - \frac{\log c}{W\left(\frac{c\alpha(x)}{\gamma(x)}\right)}
    \right)\\
    &\quad = \sum_{x\in X}\left( \alpha(x) + \alpha(x)\log W\left(\frac{c\alpha(x)}{\gamma(x)}\right) - \rho(x)\log c
    \right)
    \\
    &\quad = 1-\log c+\sum_{x\in X}\alpha(x)\log W\left(\frac{c\alpha(x)}{\gamma(x)}\right),
\end{aligned}
\end{equation}
where in the last line we use the fact that $\alpha(x)$ and $\rho(x)$ are normalized. Noting $\log W(z) = \log z-W(z)$ for $z>0$, we obtain
\begin{equation}
\begin{aligned}
    \beta \mathcal{W}_1^\mathrm{diss} &= \log \frac{e}{c} + \sum_{x\in X} \alpha(x) \left(
    \log\left(\frac{c\alpha(x)}{\gamma(x)}\right)-
    W\left(\frac{c\alpha(x)}{\gamma(x)}\right)
    \right)\\
    &= 1+D(\alpha||\gamma)- \left\langle W\left(\frac{c\alpha(x)}{\gamma(x)}\right) \right\rangle_\alpha ,
\end{aligned}
\end{equation}
where $\langle\cdot\rangle_\alpha$ is the expectation under $\alpha$. 

The maximal extracted work with $N=1$ intermediary is then $\mathcal{W}^{\mathrm{ext}}_1:=\beta^{-1}D(\alpha||\gamma)-\mathcal{W}_1^\mathrm{diss}$. It is given by
\begin{equation}
\label{eq:Wext_1}
\mathcal{W}^{\mathrm{ext}}_1=\frac{1}{\beta}\left[\left\langle W\left(\frac{c\alpha(x)}{\gamma(x)}\right) \right\rangle_\alpha-1\right]
\end{equation}
With Jensen's inequality,~\eqref{eq:Wext_1} can be upper bounded:
\begin{equation}
\label{eq:upperWext1}
\mathcal{W}^{\mathrm{ext}}_1\le \frac{1}{\beta}\left[W\left(c\sum_{x\in X}\frac{\alpha(x)^2}{\gamma(x)}\right) -1\right].
\end{equation}
 
We also seek an upper bound on $\mathcal{W}^{\mathrm{ext}}_1$ across all initial distributions $\alpha$ for a fixed final distribution $\gamma$. Note that $\sum_{x\in X}\alpha(x)^2/\gamma(x)=\left\langle \alpha(x)/\gamma(x)\right\rangle_\alpha$ is maximized at $\alpha(x)=\delta_{x,x_{\mathrm{min}}(\gamma)}$, with $x_{\mathrm{min}}(\gamma):=\mathrm{argmin}_{x\in X}\gamma(x)$, in which case the sum equals $1/\gamma(x_{\mathrm{min}})$. The value of $c$ achieved in this case is its absolute upper bound, $c=e$. Using monotonicity of $W$, we thus obtain an $\alpha$-independent upper bound:
\begin{equation}
\label{eq:upperWext1_2}
\mathcal{W}^{\mathrm{ext}}_1\le \frac{1}{\beta}\left[W\left(\frac{e}{\min_{x\in X}\gamma(x)}\right) -1\right],
\end{equation}

We can show that $\mathcal{W}^{\mathrm{ext}}_1\ge 0$ using the following Lemma.

\begin{lem}
\label{lemma_d1}
Assume $\alpha$ and $\gamma$ have full support on $X$. Let
\begin{equation}
\label{eq:Q_def}
   Q(\alpha,\gamma) := D(\alpha || \gamma) - \min_{\rho} \{D(\alpha || \rho) + D(\rho || \gamma)\}.
\end{equation}
    Then $Q(\alpha,\gamma) \ge 0$, with equality if and only if $\alpha = \gamma$.
\end{lem}

\begin{proof}
Let $\rho^{*}$ be a candidate intermediary, not necessarily optimal. Define $K(\rho^{*}) := D(\alpha || \rho^{*}) + D(\rho^{*} || \gamma)$. Then, 
\begin{equation}
\label{eq:K_rhostar_ineq}
    K(\rho^{*}) \ge \min_{\rho} \{D(\alpha || \rho) + D(\rho || \gamma)\}=D(\alpha||\gamma)-Q(\alpha,\gamma).
\end{equation}
Setting $\rho^*=\alpha$,~\eqref{eq:K_rhostar_ineq} becomes 
 \begin{equation}
\begin{aligned}
&K(\alpha) \ge D(\alpha||\gamma)-Q(\alpha,\gamma)\\
&\Rightarrow  Q(\alpha,\gamma)\ge 0,
 \end{aligned}
 \end{equation}
having used $K(\alpha)=D(\alpha||\gamma)$. If $\alpha=\gamma$, it follows from the definition~\eqref{eq:Q_def} and non-negativity of the KL divergence that $Q(\alpha,\gamma)=0$. It remains to show that $Q(\alpha,\gamma)>0$ whenever $\alpha\neq \gamma$. This is implied by showing the existence of a $\rho^*$ such that $K(\rho^*)<D(\alpha||\gamma)$. In particular, consider the linear interpolation, $\rho^{(\mu)}=\mu\gamma+(1-\mu)\alpha$, for $\mu\in[0,1]$, with $\rho^{(0)}=\alpha$, $\rho^{(1)}=\gamma$~\cite{Ay_2019}. For any $\mu\in(0,1)$, by strict convexity, 
     \begin{equation}
K(\rho_\mu) < (1-\mu)K(\alpha)+\mu K(\gamma)=D(\alpha||\gamma),
     \end{equation}
and thus there exists some minimizer $\tilde{\rho}$ with $K(\tilde{\rho})\le K(\rho_\mu)<D(\alpha||\gamma)$, for which $Q(\alpha,\gamma)>0$, unless $\alpha=\gamma$ in which case $\rho_\mu=\alpha=\gamma$.
 \end{proof}

\section{Optimal intermediaries in continuum sample spaces}
\label{ap:cont}
Proposition~\ref{prop:1} and Ref.~\cite{Pavlichin_2016} characterize the optimal intermediaries for finite sample spaces $X$; here, we consider its extension to continuous probability densities. We take $X$ as continuous and show that the optimal intermediary distributions satisfy the continuum version of~\eqref{eq:lambertsol}; for simplicity, we consider $X=\mathbb{R}$, the choice relevant to the optical trap system considered in Sec.~\ref{ssec:optical_trap}.

For a fixed choice of $\rho_0(x)$ and $\rho_{N+1}(x)$, the former assumed to be absolutely continuous with respect to the latter, the objective is to minimize the sum of continuous KL divergences between intermediaries:
\begin{equation}
    \sum_{i=1}^{N+1} D(\rho_{i-1} \| \rho_i) = \sum_{i=1}^{N+1} \int_{\mathbb{R}} \rho_{i-1}(x) \log \frac{\rho_{i-1}(x)}{\rho_i(x)}\, dx.
\end{equation}
The minimization is over all continuous probability densities $\rho_1, \ldots, \rho_N$ with sample space $\mathbb{R}$, subject to normalization $\int_{\mathbb{R}} \rho_i(x)\, dx = 1$ for each $i \in \{1, \ldots, N\}$ and absolute continuity of each $\rho_i$ with respect to $\rho_{i+1}$. We form the Lagrangian
\begin{equation}\begin{aligned}
    \mathcal{L}(\{\rho_i\}_{i=1}^N) &= \sum_{i=1}^{N+1} \int_{\mathbb{R}} \rho_{i-1}(x)\log\frac{\rho_{i-1}(x)}{\rho_i(x)}\, dx \\
    &\quad  \quad \quad - \sum_{i=1}^N \lambda_i\left(\int_{\mathbb{R}} \rho_i(x)\, dx - 1\right).
\end{aligned}\end{equation}
The terms in $\mathcal{L}$ involving $\rho_i$ are $-\int_{\mathbb{R}} \rho_{i-1}(x)\log \rho_i(x)\,dx$ from the $i$-th KL term and $\int_{\mathbb{R}} \rho_i(x)\log(\rho_i(x)/\rho_{i+1}(x))\,dx$ from the $(i+1)$-th KL term. For arbitrary function $\eta:\mathbb{R}\to \mathbb{R}$ such that $\int_{\mathbb{R}} \eta(x)dx=0$, taking the functional derivative with respect to $\rho_i$ and setting it to zero gives, for almost every $x \in \mathbb{R}$ (ignoring elements of Lebesgue measure zero), 
\begin{equation}
\begin{aligned}
\label{eq:stationarity}
    \frac{\delta\mathcal{L}}{\delta \rho} &= \frac{d}{dh}\bigg[-\int_{\mathbb{R}} \rho_{i-1}(x)\log (\rho_i(x)+h\eta(x))dx \\ &+ \int_{\mathbb{R}} (\rho_i(x)+h\eta(x))\log\left(\frac{\rho_i(x)+h\eta(x)}{\rho_{i+1}(x)}\right)dx \\&+ \lambda_i\left(\int_{\mathbb{R}} (\rho_i(x)+h\eta(x))\, dx - 1\right)\bigg]\bigg|_{h=0}\\
    = \int_{\mathbb{R}} &\eta(x)\left(1 + \lambda_i - \frac{\rho_{i-1}(x)}{\rho_i(x)} + \log\frac{\rho_i(x)}{\rho_{i+1}(x)}\right)dx=0,
\end{aligned}
\end{equation}
from which, using that $\eta(x)$ is an arbitrary variation,
\begin{equation}
\label{eq:opt_int_reals}
        1 + \lambda_i - \frac{\rho_{i-1}(x)}{\rho_i(x)} + \log\frac{\rho_i(x)}{\rho_{i+1}(x)} = 0,
\end{equation}
which is formally identical to the condition arising in the finite sample space case. Setting $u(x) := \rho_{i-1}(x)/\rho_i(x)$ and rearranging~\eqref{eq:stationarity},
\begin{equation}
    u(x) - \log u(x) = 1 + \lambda_i + \log\frac{\rho_{i-1}(x)}{\rho_{i+1}(x)}.
\end{equation}
Exponentiating both sides, multiplying through by $e^{u(x)}$, writing $c_i = e^{1+\lambda_i} > 0$, and applying the definition of the Lambert $W$ function, we arrive at
\begin{equation}
\label{eq:lambertsol_continuous}
    \rho_i(x) = \frac{\rho_{i-1}(x)}{W\!\left(\dfrac{c_i\, \rho_{i-1}(x)}{\rho_{i+1}(x)}\right)},
\end{equation}
 with constants $c_i > 0$ determined by normalization. The sum of KL divergences is strictly convex within the set of sequences of densities $\{\rho_{i}\}_{i=1}^{N}$ for which each $\rho_i$ is absolutely continuous with respect to $\rho_{i+1}$, for $i=0,...,N$, at fixed $\rho_0$ and $\rho_{N+1}$. Hence,~\eqref{eq:lambertsol_continuous} specifies the unique global minimizer; the form of~\eqref{eq:lambertsol_continuous} is formally identical to the finite sample space case. By multiplying~\eqref{eq:opt_int_reals} by $\rho_i(x)$ and integrating over $\mathbb{R}$, one obtains the relation $\lambda_i=-D(\rho_{i}||\rho_{i+1})$.

\section{Stepwise dissipation at large separation}
\label{ap:klproof}

Here we consider the limiting behavior of dissipation along each intermediate step in the 1D optical trap system considered in Sec.~\ref{ssec:optical_trap}. In particular, we characterize the two contributions $D(\rho_0||\rho_1)$ and $D(\rho_1||\pi^{H_0})$ for the case of $N=1$ optimal intermediary, as a function of the separation $x_0$ between the unimodal peaks of $\rho_0$ and $\pi^{H_0}$. Fig.~\ref{fig:klsplit} displays the result: as the separation $x_0$ increases, the KL divergence between $\rho_1$ and $\pi^{H_0}$ approaches $1$, whereas that between $\rho_0$ and $\rho_1$ (and hence the total) grows indefinitely. Accordingly, the initial step becomes the dominant contribution to dissipation. We prove this result for a general class of potentials below. In the bottom panel of Fig.~\ref{fig:klsplit}, we present numerical evidence for a similar phenomenon taking place at $N>1$; the dominant contribution still comes from the first step, with contributions from the remaining steps approaching constant values.

\begin{figure}
    \centering
    \includegraphics[width=1\linewidth]{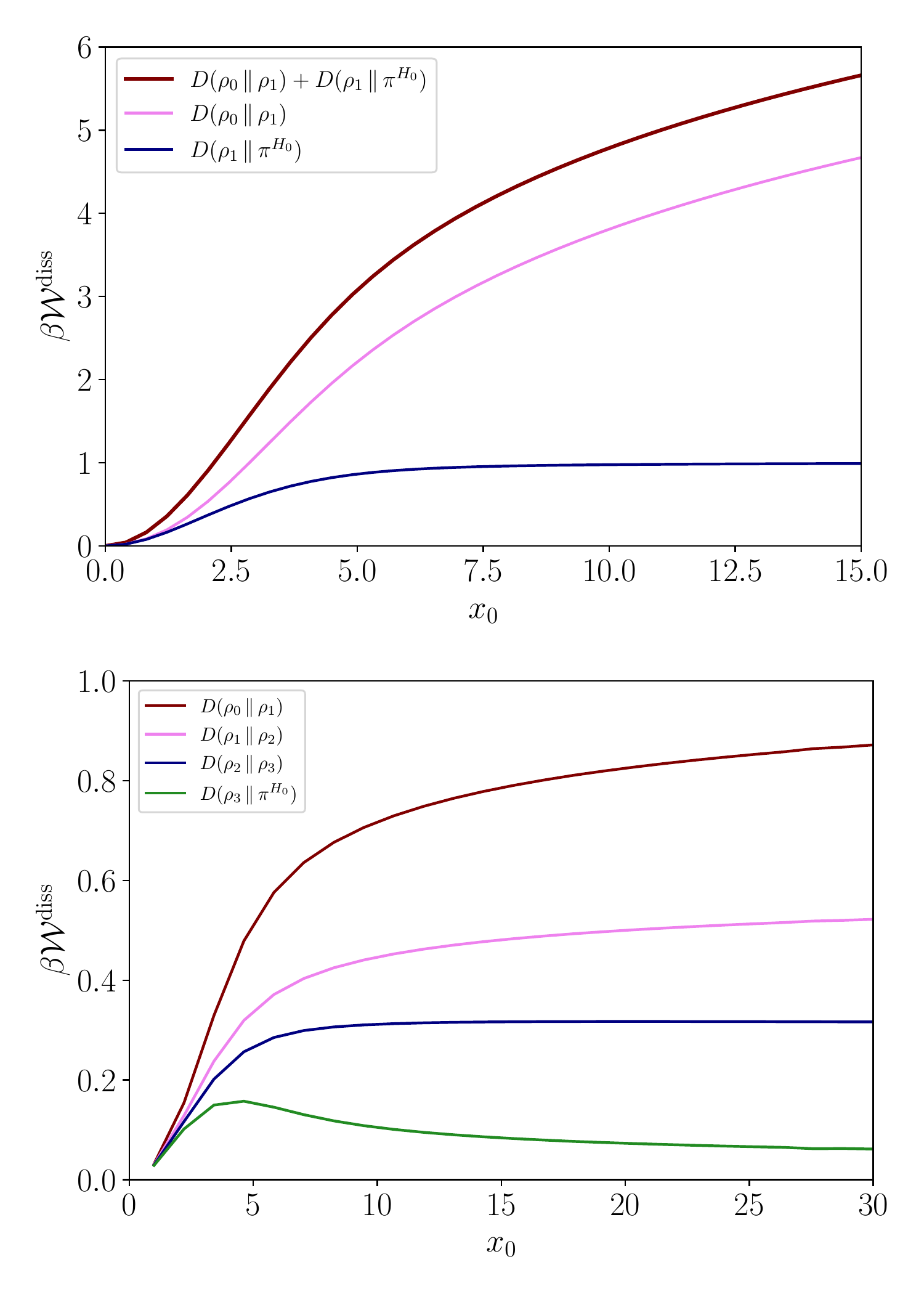}
    \caption{Top: KL divergences ($\beta \mathcal{W}^\mathrm{diss}$) between $\rho_0$, $\rho_1$ and $\pi^{H_0}$ at $\beta=1$ as a function of $x_0$ given $\rho_0(x) = e^{-x^2/2}/\sqrt{2\pi}$ and $\pi^{H_0}(x)=e^{-(x-x_0)^2/2}/\sqrt{2\pi}$. As stated in Lemma~\ref{lem:KL1_appendix}, $D(\rho_1||\pi^{H_0})\rightarrow 1$ as $x_0\rightarrow\infty$. Bottom: KL divergences between $\rho_0$, $\{\rho_i\}_{i=1}^3$ and $\pi^{H_0}$ at $\beta=1$ as a function of $x_0$. We see numerical evidence that at large $x_0$ the dominant contribution still comes from the first step, with the remaining contributions approaching constants.}
    \label{fig:klsplit}
\end{figure}

Consider distributions $\alpha(x)=e^{-\beta H_a(x)}/Z_\alpha$ and $\gamma(x) =e^{-\beta H_b(x; L)}/Z_\gamma$ parameterized by a separation parameter $L$. Define the intermediary distribution
\begin{equation}
\rho(x) = \frac{\alpha(x)}{W\left(\frac{\alpha(x)c}{\gamma(x)}\right)},
\end{equation}
where $W$ is the Lambert $W$ function and $c$ is chosen such that $\int \rho(x) dx = 1$. 

\begin{lem}
\label{lem:KL1_appendix}
Suppose there exist disjoint regions $R_1, R_2$ that partition that domain such that as $L \rightarrow \infty$:
\begin{enumerate}
\item $H_\gamma(x; L) - H_\alpha(x) \rightarrow +\infty$ for $x \in R_1$,
\item $H_\gamma(x; L) - H_\alpha(x) \rightarrow -\infty$ for $x \in R_2$,
\item $\int_{R_2} \gamma(x; L) dx \rightarrow 1$.
\item The ratio of partition functions $Z_\gamma/Z_\alpha$ and $c(L)$ are bounded above and below by positive constants. 
\end{enumerate}
Then $\mathrm{D}(\rho \| \gamma) \rightarrow 1$ as $L \rightarrow \infty$. 
\end{lem}

\begin{proof}
The KL divergence can be written as
\begin{equation}
\mathrm{D}(\rho \| \gamma) = \int_{\mathbb{R}} \rho(x) \log \frac{\rho(x)}{\gamma(x)} dx.
\end{equation}
Substituting the definition of $\rho$ and using the Lambert $W$ property $z = W(z)e^{W(z)}$, we have $\frac{\alpha c}{\gamma} = W e^W$, which gives $\log\frac{\alpha}{\gamma} = \log W + W - \log c$. Therefore,
\begin{align}
\mathrm{D}(\rho \| \gamma) &= \int_{\mathbb{R}} \frac{\alpha(x)}{W(x)} \left[\log \frac{\alpha(x)}{\gamma(x)} - \log W(x)\right] dx \\
&= \int_{\mathbb{R}} \frac{\alpha(x)}{W(x)} [W(x) - \log c] dx \\
&= 1 - \log c,
\end{align}
where we used $\int_{\mathbb{R}} \alpha(x) dx = 1$ and $\int_{\mathbb{R}} \rho(x) dx = 1$. It remains to show that $c:=c(L) \rightarrow 1$. The normalization condition is
\begin{equation}
\int \frac{\alpha(x)}{W(\alpha(x)c/\gamma(x))} dx = 1.
\end{equation}
Let $\Delta(x; L) = H_\gamma(x; L) - H_\alpha(x)$. Then the condition becomes
\begin{equation}
    \int \frac{\alpha(x)}{W((Z_\gamma/Z_\alpha) ce^{\beta\Delta(x;L)})} dx = 1.
\end{equation}
Consider the limit of $L\rightarrow \infty$. Note that $R_1,R_2$ partition the domain (up to a measure zero set), so we can split this integral into its contributions over disjoint regions $R_1,R_2$.
\begin{equation}
    \int_{R_1} \frac{\alpha(x)}{W\left(\frac{Z_\gamma}{Z_\alpha} ce^{\beta\Delta(x;L)}\right)} dx + \int_{R_2} \frac{\alpha(x)}{W\left(\frac{Z_\gamma}{Z_\alpha} ce^{\beta\Delta(x;L)}\right)} dx= 1.
\end{equation}
In region $R_1$ where $\Delta\to\infty$, assumption (4) implies
\begin{equation}
W\!\left(
\frac{Z_\gamma}{Z_\alpha}
ce^{\beta\Delta}
\right)\rightarrow\infty.
\end{equation}
Hence
\begin{equation}
\frac{\alpha(x)}
{W((Z_\gamma/Z_\alpha)e^{\beta\Delta}c)}
\rightarrow 0,
\end{equation}
and in particular,
\begin{equation}
    0 \leq \frac{\alpha(x)}{W\left(\frac{Z_\gamma}{Z_\alpha} ce^{\beta\Delta(x;L)}\right)}\leq \alpha(x),
\end{equation}
on $R_1$ as $L\rightarrow \infty$. Then, since $\alpha(x)$ is integrable, it follows by dominated convergence that
\begin{equation}
    \int_{R_1} \frac{\alpha(x)}{W\left(\frac{Z_\gamma}{Z_\alpha} ce^{\beta\Delta(x;L)}\right)} dx \rightarrow 0.
\end{equation}

In region $R_2$ where $\Delta\rightarrow-\infty$, we have
\begin{equation}
    W\left(
\frac{Z_\gamma}{Z_\alpha}
ce^{\beta\Delta}
\right)
\sim
\left(\frac{Z_\gamma}{Z_\alpha}\right)
ce^{\beta\Delta},
\end{equation}
so that
\begin{equation}
\frac{\alpha(x)}
{W\left(\frac{Z_\gamma}{Z_\alpha} ce^{\beta\Delta(x;L)}\right)} \sim \frac{\alpha(x)}
{c(\alpha(x)/\gamma(x))}
=
\frac{\gamma(x)}{c},
\end{equation}
where we substituted $c\alpha(x)/\gamma(x)$ for $(Z_\gamma/Z_\alpha)ce^{\beta\Delta(x;L)}$.
Thus the normalization condition becomes
\begin{equation}
\frac{1}{c} \int_{R_2} \gamma(x) dx + o(1) = 1.
\end{equation}
By assumption (3), $\int_{R_2} \gamma(x) dx \rightarrow 1$, hence $c \rightarrow 1$ and $\mathrm{D}(\rho \| \gamma) \rightarrow 1$.
\end{proof}

We note that this result applies to a broad class of potentials including single-well, double-well, and mixture systems with growing separation.

\section{Relaxation time in enzymatic reaction}
\label{ap:relaxation_time}

Here we derive an estimate of the equilibration timescale $\mathcal{T}^{\mathrm{equi}}$ invoked in Sec.~\ref{ssec:enzyme}. Namely, we approximate $\mathcal{T}^{\mathrm{equi}}$ by the relaxation time $\tau$ of a continuous-time Markov chain rate with probability rate matrix determined by microstate energies. Let $\epsilon_i$ be the energy of microstate $i$, for $i \in \{B, I, U\}$. Let $k_{ij}$ denote the rate of transition to microstate $i$ from microstate $j$. Assuming Arrhenius rates~\cite{zwanzig2001nonequilibrium}, we have

\begin{equation}
\begin{aligned}
    k_{IB} &= \nu e^{-\beta(\epsilon_I - \epsilon_B)}, \quad k_{BI} = \nu, \\
    k_{IU} &= \nu e^{-\beta(\epsilon_I - \epsilon_U)}, \quad k_{UI} = \nu,\\
    k_{UB} &= 0 , \quad\quad\quad \quad \quad \ k_{BU}=0,
\end{aligned}
\end{equation}

where $\nu$ is a constant (possibly depending on $\beta$ \cite{arcus2020}) affecting the overall rate. Note that $k_{BU}=k_{UB}=0$ follows from assigning zero kinetic activity to transitions between microstates $B$ and $U$, rather than from their energies. Let $p_i(t)$ denote the probability that the microstate is $i$ at time $t$. The corresponding master equation is:
\begin{equation}
\label{eq_master}
    \frac{d}{dt} \begin{bmatrix}
        p_B \\
        p_I \\
        p_U
    \end{bmatrix} = 
    \begin{bmatrix}
        -k_{IB} & k_{BI} & 0 \\
        k_{IB} & -(k_{BI}+k_{UI}) & k_{IU} \\
        0 & k_{UI} & -k_{IU} 
    \end{bmatrix}
    \begin{bmatrix}
        p_B \\
        p_I \\
        p_U
    \end{bmatrix}.
\end{equation}

The mean residence time in each microstate is the inverse of its total exit rate: microstate $I$ has mean residence time $1/(k_{BI}+k_{UI}) = 1/2\nu$, while microstate $B$ and $U$ have mean residence time $1/k_{IB} = e^{\beta (\epsilon_I-\epsilon_B)}/\nu$ and $1/k_{IU} = e^{\beta (\epsilon_I-\epsilon_U)}/\nu$, respectively. When the intermediate microstate $I$ has high enough energy, i.e., $\epsilon_I \gg \epsilon_B, \epsilon_U$, the total exit rate of microstate $I$ is much faster compared to exit rates of $B$ or $U$. This introduces a separation of timescales: $p_B$ and $p_U$ change at the slow timescale $\sim 1/k_{IB}$ (or $1/k_{IU}$), and $p_I$ changes at faster timescale of $\sim 1/2\nu$. In this way, on the slow timescale, $p_I$ is determined by the instantaneous values of $p_B$ and $p_U$, and we may set $\dot{p}_I = 0$. The instantaneous probability of microstate $I$, denoted under this approximation by $p^*_I(t)$, is then given by

\begin{equation}
\begin{aligned}
    0 &= k_{IB}p_B(t) -(k_{BI}+k_{UI})p^*_I + k_{IU}p_U(t)  \\
&\Rightarrow  p^*_I(t) = \frac{k_{IU}p_U(t)+k_{IB}p_B(t)}{2\nu}.
\end{aligned}
\end{equation}

Substituting this into Eq.~(\ref{eq_master}), we obtain an effective master equation: 

\begin{equation}
    \frac{d}{dt} \begin{bmatrix}
        p_B\\p_U
    \end{bmatrix} =  
    \begin{bmatrix}
    -k_{IB}/2 & k_{IU}/2 \\
    k_{IB}/2 & -k_{IU}/2
    \end{bmatrix} \begin{bmatrix}
        p_B\\p_U
    \end{bmatrix}.
\end{equation}
The the smallest nonzero eigenvalue of this effective rate matrix is $\alpha =  -(k_{IB} + k_{IU})/2$, corresponding to the slowest-decaying nonstationary eigenvector. The `relaxation time' is defined as the reciprocal of $|\alpha|$ \cite{levin2017markov}:
\begin{align}
    \tau &= \frac{1}{|\alpha|}= \frac{2}{k_{IB}+k_{IU}} = \frac{2}{\nu}\frac{e^{\beta\epsilon_I}}{e^{\beta\epsilon_B} + e^{\beta\epsilon_U}}.
\end{align}
One can further simplify $\tau$ by considering either 
\begin{itemize}
    \item[(a)] $\epsilon_B\gg\epsilon_U$, which results in
    $\tau = (2/\nu)e^{\beta(\epsilon_I - \epsilon_B)}$, or 
    \item[(b)] $|\epsilon_B-\epsilon_U|\ll 1$, which results in $\tau = (1/\nu)e^{\beta(\epsilon_I - \epsilon_B)}$.
\end{itemize}
In both cases, 
\begin{equation}
\label{eq:tau}
    \tau = \mathrm{const}\times e^{\beta(\epsilon_I - \epsilon_B)},
\end{equation}
with the constant having a possible $\beta$-dependence through $\nu$. The form~\eqref{eq:tau} can also be obtained as a limiting regime of the equilibration time associated with the full $3\times 3$ system:
\begin{equation}
\tau=\frac{2 e^{\epsilon_I \beta}/\nu}{e^{\epsilon_B \beta} + 2 e^{\epsilon_I \beta} + e^{
   \epsilon_U \beta}- \varrho},
\end{equation}
where
\begin{equation}
    \varrho=\sqrt{
   e^{2 \beta\epsilon_B } + 4 e^{2 \beta\epsilon_I } + e^{2 \beta\epsilon_U} - 
    2 e^{\beta(\epsilon_B + \epsilon_U)}}.
\end{equation}
In the enzyme model of Sec.~\ref{ssec:enzyme}, under the effective quench, $\epsilon_I$ is replaced by $\epsilon_I'=\epsilon_I-\delta$. Then~\eqref{eq:tau} becomes equivalent to~\eqref{eq:Tequi}. With $\epsilon_B=1$, $\epsilon_U=0$ and $\nu=1$, the full expression for $\mathcal{T}^{\mathrm{equi}}$ reads
\begin{equation}
\mathcal{T}^{\mathrm{equi}}=\frac{2 e^{\beta\epsilon_I'}}{1+e^{\beta} + 2 e^{\beta\epsilon_I'}- \sqrt{1+
   e^{2 \beta}- 
    2 e^{\beta} + 4 e^{2 \beta\epsilon_I'}  }}.
\end{equation}

\section{Estimation of $\lambda$ in enzyme model}
\label{ap:estimation}

Here we illustrate that the parameter $\lambda$ appearing in the cost function $C(\delta)$ of Sec.~\ref{sssec:cost} can be estimated from data, under the assumption that such data arose via the minimization of $C(\delta)$. At a minimum,
\begin{equation}
\label{eq:Cmax}
\frac{dC(\delta)}{d\delta}=\frac{d\mathcal{W}^{\mathrm{diss}}(\delta)}{d\delta}+\lambda\frac{d\mathcal{T}^{\mathrm{equi}}(\delta)}{d\delta}=0.
\end{equation}
Denoting the solution of~\eqref{eq:Cmax} as $\delta=\delta^*$, the parameter $\lambda$ can be obtained by simple rearrangement:
\begin{equation}
\label{eq:lam_est}
\lambda=-\left.\frac{d\mathcal{W}^{\mathrm{diss}}(\delta)}{d\delta}\left/\frac{d\mathcal{T}^{\mathrm{equi}}(\delta)}{d\delta}\right.\right\vert_{\delta=\delta^*}
\end{equation}
In the enzyme example of Sec.~\ref{ssec:enzyme}, the right-hand side of~\eqref{eq:lam_est} before setting $\delta=\delta^*$ is a function of energies $(\epsilon_B,\epsilon_I,\epsilon_U)$, inverse temperature $\beta$ and the value of function $M(\beta)$, as well as the change in barrier height $\delta$. We may observe an experimentally measured value $\delta=\hat\delta$, and make the assumption of optimality, i.e., take $\delta^*=\hat{\delta}$. Then, given estimates of the other parameters, one may compute~\eqref{eq:lam_est} directly. In the enzyme model specifically, at $\epsilon_B=1$ and $\epsilon_U=0$, we have
\begin{equation}
\frac{d\mathcal{W}^{\mathrm{diss}}(\delta)}{d\delta}=\frac{d}{d\delta}\left(\frac{\delta e^{-\beta(\epsilon_I-\delta)}}{Z'(\delta)}\right)=\pi^{H'}_I[1+\delta\beta(1-\pi^{H'}_I)]
\end{equation}
and
\begin{equation}
\frac{d\mathcal{T}^{\mathrm{equi}}(\delta)}{d\delta}=\frac{d}{d\delta}\left(M(\beta)e^{\beta-\delta-\epsilon_B}\right)=-\beta \mathcal{T}^{\mathrm{equi}}(\delta),
\end{equation}
from which, given estimates of $\beta,\epsilon_I$, and either $\mathcal{T}^{\mathrm{equi}}$ or $M(\beta)$, and assuming optimality of the measured value $\delta=\hat\delta$,~\eqref{eq:lam_est} provides an estimate $\hat\lambda$ of $\lambda$:
\begin{equation}
\hat\lambda=\frac{\pi^{H'}_I(\hat\delta)[1+\hat\delta\beta(1-\pi^{H'}_I(\hat\delta))]}{\beta \mathcal{T}^{\mathrm{equi}}(\hat\delta)},
\end{equation}
where $\pi^{H'}_I(\hat\delta):=e^{-\beta (\epsilon_I-\hat\delta)}/Z'(\hat\delta)$.

\end{document}